\documentclass[a4paper, english, pagebackref, runningheads]{article}
\usepackage[utf8]{inputenc}
\usepackage{amsmath}
\usepackage{amsfonts}
\usepackage{amssymb}
\usepackage{amsthm}
\usepackage{xcolor,graphicx}
\usepackage{tabularx}
\usepackage{booktabs}
\usepackage{pdflscape}
\usepackage{xspace}
\usepackage{a4wide}
\usepackage{enumerate}
\usepackage{paralist}
\usepackage{authblk}

\usepackage{rotating}
\usepackage[numbers,sort,sectionbib]{natbib}

\usepackage[vlined,linesnumbered,ruled]{algorithm2e}
\SetEndCharOfAlgoLine{}
\SetCommentSty{textrm}

\usepackage{tikz}
\usetikzlibrary{matrix,arrows,decorations.pathmorphing,decorations.pathreplacing,backgrounds,positioning,fit,shapes,arrows,calc,patterns}

\usepackage[pagebackref,pdfdisplaydoctitle,menucolor=orange!40!black,filecolor=magenta!40!black,urlcolor=blue!40!black,linkcolor=red!40!black,citecolor=green!40!black,colorlinks]{hyperref}
\hypersetup{pdftitle={Hyper}, pdfauthor={Mertzios et al.\ }}

\usepackage[color=green!20,disable]{todonotes}
\usepackage[sort&compress,nameinlink,noabbrev,capitalize]{cleveref}

\theoremstyle{plain}
\newtheorem{theorem}{Theorem}[section]
\newtheorem{lemma}[theorem]{Lemma}
\newtheorem{claim}[theorem]{Claim}

\newtheorem{observation}[theorem]{Observation}
\newtheorem{rrule}{Reduction Rule}[section]
\theoremstyle{definition}
\newtheorem{definition}[theorem]{Definition}

\crefname{cond}{Condition}{Conditions}
\crefname{claim}{Claim}{Claim}

\crefname{rrule}{Reduction Rule}{Reduction Rules}
\crefname{step}{Step}{Steps}

\newcommand{\problemdef}[3]{
	\begin{center}
		\begin{minipage}{0.98\textwidth}
			\textsc{#1}\smallskip\\
			\setlength{\tabcolsep}{3pt}
			\begin{tabularx}{\textwidth}{@{}lX@{}}
				\textbf{Input:} 	& #2 \\
				\textbf{Question:} 	& #3
			\end{tabularx}
		\end{minipage}
	\end{center}
}

\newcommand{\dist}[2]{\ensuremath{\overline{#1 #2}}}
\newcommand{\distP}[3]{\ensuremath{\overline{#1 #2}|_{#3}}}

\newcommand{\hyp}{\textsc{Hyperbolicity}\xspace}
\newcommand{\NP}{\ensuremath{\mathsf{NP}}}

\usepackage{etoolbox}

\newcommand{\appsymb}{$\star$}
\newcommand{\appref}[1]{{\hyperref[proof:#1]{\appsymb}}}

\newcommand{\raproof}{($\Rightarrow$)}
\newcommand{\laproof}{($\Leftarrow$)}

\renewcommand{\vec}[1]{\overrightarrow{#1}}
\newcommand{\ceil}[1]{\lceil#1\rceil}

\title{When can Graph Hyperbolicity be computed in~Linear~Time?\thanks{This work was initiated at the 2016 research retreat of the Algorithmics and Computational Complexity (AKT) group of TU Berlin.}}

\author[1]{Till Fluschnik\thanks{Supported by the DFG, project DAMM (NI~369/13-2).}}
\author[2]{Christian Komusiewicz\thanks{Supported by the DFG, project MAGZ (KO 3669/4-1).}}
\author[3]{George B. Mertzios} 
\author[1,3]{Andr\'{e}~Nichterlein\thanks{Supported by a postdoc fellowship of the DAAD while at Durham University.}} 
\author[1]{Rolf Niedermeier} 
\author[4]{Nimrod Talmon\thanks{Nimrod Talmon was supported by a postdoctoral fellowship from I-CORE ALGO.}}
\affil[1]{\small{Institut f\"ur Softwaretechnik und Theoretische Informatik, TU~Berlin, Germany, 

\texttt{\{till.fluschnik, andre.nichterlein, rolf.niedermeier\}@tu-berlin.de}}}

\affil[2]{Friedrich-Schiller-Universität Jena, Germany, 

\texttt{christian.komusiewicz@uni-jena.de}}

\affil[3]{School of Engineering and Computing Sciences, Durham University, UK, 

\texttt{george.mertzios@durham.ac.uk}}

\affil[4]{Weizmann Institute of Science, Rehovot, Israel, 

\texttt{nimrodtalmon77@gmail.com}}

\date{}
\begin{document}

\maketitle

\begin{abstract}
Hyperbolicity measures, in terms of (distance) metrics, how close a given graph
is to being a tree. Due to its relevance in modeling real-world networks,
hyperbolicity has seen intensive research over the last years.
Unfortunately, the best known algorithms for computing the hyperbolicity number
of a graph (the smaller, the more tree-like) have running time~$O(n^4)$, 
where $n$ is the number of graph vertices. 
Exploiting the framework of parameterized complexity analysis, we 
explore possibilities for ``linear-time FPT'' algorithms to 
compute hyperbolicity. For instance, we show that hyperbolicity can be computed in time $O(2^{O(k)} + n +m)$ ($m$ being the number of graph edges) 
 while at the same time, unless the SETH fails, there 
is no $2^{o(k)}n^2$-time algorithm.

\end{abstract}

\section{Introduction}
(Gromov) hyperbolicity~\cite{Gro87} of a graph 
is a popular attempt to capture and 
measure how \emph{metrically} close a graph is to being a tree.
The study of hyperbolicity is motivated by the fact that many real-world 
graphs are tree-like from a distance metric point of view~\cite{AD16,BCCM15}. 
This is due to the fact that many of these graphs (including 
Internet application networks or social networks) 
possess certain geometric and topological characteristics. 
Hence, for many applications,
including the design of (more) 
efficient algorithms,
it is useful to know the hyperbolicity of a graph.
The hyperbolicity of a graph is a nonnegative number~$\delta$;
the smaller~$\delta$ is,
the more tree-like the graph is; in particular, $\delta =0$ means that the graph metric indeed is a tree 
metric.
Typical hyperbolicity values for real-world graphs 
are below~5~\cite{AD16}.  

Hyperbolicity can be defined via a four-point condition:
Considering all size-four subsets $\{ a, b, c, d \}$ 
of the vertex set of the graph, 
one takes the  (nonnegative) difference 
between the biggest two of the three sums
$\dist{a}{b}+\dist{c}{d}$, $\dist{a}{c}+\dist{b}{d}$, and 
$\dist{a}{d}+\dist{b}{c}$,
where, e.g., $\dist{a}{b}$ denotes the length 
of the shortest path between vertices $a$ and~$b$ in the given graph.
For an $n$-vertex graph, this characterization of 
hyperbolicity directly implies a simple (brute-force)
$O(n^4)$-time algorithm to compute its hyperbolicity.
It has been observed that this polynomial running time is too slow 
for computing the hyperbolicity of big graphs as occurring in 
applications~\cite{AD16,BCCM15,BCH16,FIV15}. 
On the theoretical side, it was shown that relying on some (rather impractical)
matrix multiplication results, one can improve the upper bound to
$O(n^{3.69})$~\cite{FIV15}. Moreover, roughly quadratic lower bounds are 
known~\cite{BCH16,FIV15}. In practice, however, the best known algorithm
still has an $O(n^4)$-time worst-case bound but uses several 
clever tricks when compared to the straightforward brute-force 
algorithm~\cite{BCCM15}. Indeed, based on empirical studies 
an $O(mn)$~running time is claimed, where $m$~is the number
of edges in the graph. 

To explore the possibility of faster algorithms for hyperbolicity in 
relevant special cases is the guiding principle of this work.
More specifically, introducing some graph parameters, we investigate 
whether one can compute hyperbolicity in linear time when these parameters take small values. In other words, we employ the framework of parameterized 
complexity analysis (so far mainly used for studying \NP-hard problems) 
applied to the polynomial-time solvable hyperbolicity problem. 
In this sense, we follow the recent trend of studying 
``FPT~in~P''~\cite{GMN15}. Indeed, other than for \NP-hard problems, for some 
parameters we 
achieve not only exponential dependence on the parameter but also polynomial
ones. 
\paragraph{Our contributions.}
\cref{tab:results} summarizes our main results.
On the positive side, for a number of natural graph parameters we can
attain ``linear FPT'' running times. Our ``positive'' graph parameters 
here are the following:
\begin{compactitem}
\item the covering path number, that is, the minimum number of paths where only 
the endpoints have degree greater than two and which cover all vertices;
\item the feedback edge number, that is, the minimum number of edges to 
delete to obtain a forest;
\item the number of graph vertices of degree at least three; 
\item the vertex cover number, that is the minimum number of 
vertices needed to cover all edges in the graph;
\item the minimum vertex deletion number to cographs, that is, the minimum number of vertices to delete to obtain a cograph.\footnote{Cographs are the graphs without induced~$P_4$s. For instance, distance to cographs is
never bigger than the graph parameter cluster graph vertex deletion distance~\cite{DK12}. Moreover, it is also upper-bounded by the vertex cover number.}
\end{compactitem}

On the negative side, we prove that that with respect to the parameter vertex 
cover number~$k$, we cannot hope for any $2^{o(k)} n^{2-\epsilon}$ algorithm 
unless the SETH fails. We also obtain a ``quadratic-time FPT''~lower bound
with respect to the parameter maximum vertex degree, again assuming SETH.
Finally, we show that computing the hyperbolicity is at least as hard as computing a size-four independent set of a graph. 
It is conjectured that computing size-four independent sets needs~$\Omega(n^3)$ time.

\begin{table}[t]
  \centering
 \caption{Summary of our algorithmic results. Herein, $k$~denotes the parameter and~$n$ 
    and $m$ denote the number of vertices and edges, respectively.}
  \begin{tabular}{lrr}\toprule
  Parameter & \multicolumn{1}{l}{Running time} \\\midrule
  covering path number & $O(k^4(n+m))$ & [\cref{thm:maximal-paths}] \\
  feedback edge number & $O(k^4(n+m))$ & [\cref{thm:fes}] \\
  number of $\geq 3$-degree vertices & $O(k^8(n+m))$ & [\cref{thm:degree3vertices}]\\
  vertex cover number & $2^{O(k)} + O(n+m)$ & [\cref{theorem:vc}]\\
  distance to cographs & $O(4^{4k}\cdot k^7\cdot (n+m))$ &  [\cref{thm:cograph-dist}]\\
   \bottomrule
  \end{tabular}

 \label{tab:results}
\end{table}

\section{Preliminaries and Basic Observations}

We write $[n]:=\{1,\ldots,n\}$ for every $n\in\mathbb N$. For a function~$f:X\to Y$
and~$X'\subseteq X$ we set~$f(X'):=\{y\in Y\mid\exists x\in X':f(x)=y\}$.
\paragraph{Graph theory.} 
Let $G=(V,E)$~be a graph.  We define $|G|=|V|+|E|$.  For
$W\subseteq V$, we denote by~$G[W]$ the graph \emph{induced} by~$W$. We
use~$G-W:=G[V\setminus W]$ to denote the graph obtained from~$G$ by deleting the vertices
of~$W \subseteq V$.
A \emph{path} $P=(v_1,\ldots,v_k)$ in~$G$ is a tuple of distinct vertices in~$V$ such
that $\{v_i,v_{i+1}\}\in E$ for all $i\in[k-1]$;  we say that such a path $P$ has endpoints $v_1$ and
$v_k$, we call the other vertices of $P$ (i.e., $P \setminus \{v_1, v_k\}$) as inner nodes,
and we say that $P$ is a $v_1$-$v_k$~path.
We denote by $\dist{a}{b}$ the length of a
shortest~$a$-$b$ path if such a path exists; otherwise, that is, if~$a$ and~$b$ are in
different connected components, $\dist{a}{b}:=\infty$.  Let~$P=(v_1,\ldots,v_k)$ be a path
and~$v_i,v_j$ two vertices on~$P$.  We denote by~$\distP{v_i}{v_j}{P}$ the distance
of~$v_i$ to~$v_j$ on~$P$, that is, $\distP{v_i}{v_j}{P} = |j-i|$.
For a graph~$G$ we denote with~$V^{\ge 3}_{G}$ the set of vertices of~$G$ that have degree at least three.

\paragraph{Hyperbolicity.}
Let $G=(V,E)$ be graph and $a,b,c,d\in V$.
We denote the distance between two vertices~$a$ and~$b$ by \dist{a}{b}.
We define $D_1:=\dist{a}{b}+\dist{c}{d}$, $D_2:=\dist{a}{c}+\dist{b}{d}$, and $D_3:=\dist{a}{d}+\dist{b}{c}$ (referred to as \emph{distance sums}).
Moreover, we define $\delta(a,b,c,d) := |D_i-D_j|$ if $D_k\le\min\{D_i,D_j\}$, for pairwise distinct $i,j,k\in\{1,2,3\}$.
The \emph{hyperbolicity} of a graph is defined as
$\delta(G) = \max_{a,b,c,d \in V}\{\delta(a,b,c,d)\}$.
We say that the graph is \emph{$\delta$-hyperbolic} for some~$\delta \in \mathbb N$ if it has hyperbolicity at most~$\delta$.
That is, a graph is~$\delta$-hyperbolic if for each 4-tuple $a,b,c,d \in V$ we have
$$ \dist{a}{b} + \dist{c}{d} \le \max\{\dist{a}{c}+\dist{b}{d},\dist{a}{d}+\dist{b}{c}\} + \delta.$$
Formally, the \hyp problem is defined as follows.

\problemdef{\hyp}
        {An undirected graph~$G=(V,E)$ and a positive integer~$\delta$.}
        {Is~$G$ $\delta$-hyperbolic?}

The following lemmas would be useful later.
For any quadruple $\{a,b,c,d\}$,
\cref{lem:hyp-distance-bounded} upper bounds $\delta(a, b, c, d)$
by twice the distance between any pair of vertices of the quadruple.
\cref{lem:diamequalshyp} discusses graphs for which the hyperbolicity equals the diameter.
\cref{lemma:1seps} is used in the proof of \cref{rrule:degree1}.

\begin{lemma}[{\cite[Lemma 3.1]{CCL15}}] \label{lem:hyp-distance-bounded}
	$\delta(a,b,c,d) \le 2 \cdot \min_{u \neq v \in \{a,b,c,d\}} \{\dist{u}{v}\}$
\end{lemma}

\begin{lemma}%
  \label{lem:diamequalshyp}
 Let $G$ be a graph with diameter $h$ and $\delta(G)=h$.
 Then for each quadruple $a,b,c,d\in V(G)$ with $\delta(a,b,c,d)=h$, it holds that exactly two disjoint pairs are at distance~$h$ and all the other pairs are at distance~$h/2$.
\end{lemma}

{
  \begin{proof}
  Let $a,b,c,d\in V(G)$ be an arbitrary but fixed quadruple with $\delta(a,b,c,d)=h$.
  By \cref{lem:hyp-distance-bounded}, $\min_{u \neq v \in \{a,b,c,d\}} \{\dist{u}{v}\}\geq h/2$.
  Let w.l.o.g. be $S_1=\dist{a}{b}+\dist{c}{d}$ and $S_1\geq \max\{S_2,S_3\}$.
  Then $h=S_1-\max\{S_2,S_3\}\leq S_1-h$.
  It follows that $S_1\geq 2h$ and since $G$ is of diameter $h$, it follows that $\dist{a}{b}=\dist{c}{d}=h$.
  Moreover, it follows that~$\max\{S_2,S_3\}=h$ and together with $\min_{u \neq v \in \{a,b,c,d\}} \{\dist{u}{v}\}\geq h/2$, we obtain that each other distance equals $h/2$.
  \end{proof}
}

\begin{lemma}%
  \label{lemma:1seps}
 Given a graph $G=(V,E)$ with $|V|>4$ and $v\in V$ being a 1-separator in~$G$. 
 Let $A_1,\ldots,A_\ell$ be the components in $G-\{v\}$. 
 Then there is an $i\in[\ell]$ such that $\delta(G)=\delta(G-V(A_i))$.
\end{lemma}

{
  \begin{proof}
  Let $A_i$ be one of the components in $G-\{v\}$ with $\delta(G[A_i\cup\{v\}])$ being minimum if $|V(A_j)\cup\{v\}|\geq 4$ for all $j\in[\ell]$, or with $|V(A_i)|$ being minimum otherwise.
  We distinguish three cases. 
  Let $a,b,c,d\in V$ such that (we assume $|V(A_1)|\geq 3$) either
  \begin{enumerate}[(i)]
    \item $a,b,c\in V\backslash (V(A_i)\cup\{v\})$ and $d\in V(A_i)$, or
    \item $a,b\in V\backslash (V(A_i)\cup\{v\})$ and $c,d\in V(A_i)$ (assuming $|V(A_i)|\geq 2$), or
    \item $a,b\in V\backslash (V(A_i)\cup\{v\})$, $c=v$, and $d\in V(A_i)$.
  \end{enumerate}

  \paragraph{Case (i):} 
  In this case, every shortest path from $d$ to any of $a,b,c$ contains $v$. Hence, we obtain
  \begin{align*}
  \dist{a}{b}+\dist{c}{d} &= \dist{a}{b}+\dist{v}{c}+\dist{v}{d}, \\
  \dist{a}{c}+\dist{b}{d} &= \dist{a}{c}+\dist{v}{b}+\dist{v}{d}, \\
  \dist{a}{d}+\dist{b}{c} &= \dist{b}{c}+\dist{v}{a}+\dist{v}{d},
  \end{align*}
  and thus $\delta(a,b,c,d)=\delta(a,b,c,v)$.

  \paragraph{Case (ii):} 
  In this case, every shortest path between $a,b$ and $c,d$ contains $v$. Hence, we obtain
  \begin{align*}
  \dist{a}{b}+\dist{c}{d} &= \dist{a}{b}+\dist{c}{d}, \\
  \dist{a}{c}+\dist{b}{d} &= \dist{a}{v}+\dist{v}{c}+\dist{v}{b}+\dist{v}{d}, \\
  \dist{a}{d}+\dist{b}{c} &= \dist{a}{v}+\dist{v}{c}+\dist{v}{b}+\dist{v}{d}.
  \end{align*}
  Since $\dist{a}{b}\leq \dist{a}{v}+\dist{v}{b}$ on the one hand, and $\dist{c}{d}\leq \dist{c}{v}+\dist{v}{d}$ on the other hand, it follows that $\delta(a,b,c,d)=0$.

  \paragraph{Case (iii):} 
  In this case, $c$ is contained in every shortest path. Hence, we obtain
  \begin{align*}
  \dist{a}{b}+\dist{c}{d} &= \dist{a}{b}+\dist{c}{d}, \\
  \dist{a}{c}+\dist{b}{d} &= \dist{a}{c}+\dist{b}{c}+\dist{c}{d}, \\
  \dist{a}{d}+\dist{b}{c} &= \dist{a}{c}+\dist{c}{d}+\dist{b}{c}.
  \end{align*}
  Since $\dist{a}{b}\leq \dist{a}{c}+\dist{c}{b}$, it follows that $\delta(a,b,c,d)=0$.

  Observe that the case where $a\in V\backslash (V(A_i)\cup\{v\})$, $c=v$, and $b,d\in V(A_i)$ reduces to~(iii).
  Since $A_i$ was chosen as $\delta(G[A_i\cup\{v\}])$ being minimum if $|V(A_j)\cup\{v\}|\geq 4$ for all $j\in[\ell]$, or with $|V(A_i)|$ being minimum otherwise, it follows that $\delta(G)=\delta(G-V(A_i))$.
  \end{proof}
}

\begin{rrule}\label{rrule:degree1}
 As long as there are more than four vertices, remove vertices of degree one.
\end{rrule}

\begin{lemma}%
  \label{lem:rrule-1-correct-lin-time}
	\cref{rrule:degree1} is correct and can be applied exhaustively in linear time.
\end{lemma}
{
\begin{proof}
	The soundness of \cref{rrule:degree1} follows immediately from \cref{lemma:1seps}.
	To apply \cref{rrule:degree1} in linear time do the following.
 	First, collect all degree one vertices in linear time in a list~$L$.
	Then, iteratively delete degree-one vertices and put their neighbor in~$L$ if it has degree one after the deletion.
	Each iteration can be applied in constant time.
	Thus, \cref{rrule:degree1} can be applied in linear time.
\end{proof}
}

\section{Polynomial Linear-Time Parameterized Algorithms}
In this section, we provide polynomial linear-time parameterized algorithms with respect to the parameters feedback edge number and number of vertices with degree at least three;
that is, algorithms with a linear-time dependence on the input size times a polynomial-time dependence on the parameter value.

To this end, we first introduce an auxiliary parameter, the \emph{minimum maximal paths cover number},
which we formally define below and also describe a polynomial linear-time paramaterized algorithm for it.

Building upon this result,
for the parameter feedback edge number we then show that,
after applying \cref{rrule:degree1}, the number of maximal paths can be upper bounded by a polynomial of the feedback edge number.
This implies a polynomial linear-time parameterized algorithm for the feedback edge number as well.
For the parameter number of vertices with degree at least three,
we introduce an additional reduction rule to achieve that the number of maximal paths is bounded in a polynomial of this parameter.
Again, this implies a polynomial linear-time algorithm.

\paragraph{Minimum maximal paths cover number.}

Consider the following definition.

\begin{definition}[Maximal path]
Let $G$ be a graph and $P$ be a path in $G$.
Then,
$P$ is a \emph{maximal path} if the following hold:
(1) it contains at least two vertices;
(2) all its inner nodes have degree two in $G$;
and
(3) either both its endpoints have degree at least three in $G$,
or one of its endpoints has degree at least three in~$G$
while the other endpoint is of degree two in $G$;
and
(4) $P$ is size-wise maximal with respect to these properties.
\end{definition}

We will be interested in the minimum number of maximal paths needed to cover the vertices of a given graph;
we call this number the \emph{minimum maximal paths cover number}.
While not all graphs can be covered by maximal paths
(e.g., edgeless graphs),
graphs which have minimum degree two and contain no isolated cycles can be covered by maximal paths
(it follows by, e.g., a greedy algorithm which iteratively selects an arbitrary uncovered vertex and exhaustively extend it arbitrarily;
since there are no isolated cycles and the minimum degree is two,
we are bound to eventually hit at least one vertex of degree three).
In the following lemma we show how to approximate the minimum maximal paths cover number,
for graphs which have minimum degree two and contain no isolated cycles.

\begin{lemma}%
  \label{lem:maximal-paths}
  There is a linear time algorithm which approximates the minimum maximal paths cover number
  for graphs which have minimum degree two and contain no isolated cycles.
\end{lemma}

{
\begin{proof}
The algorithm operates in two phases.
In the first phase,
we greedily cover all vertices of degree two.
Specifically,
we arbitrarily select a vertex of degree two,
view it as a path of length one,
and arbitrarily try to extend it in both directions
(it has degree two, so, pictorially, has two possible directions for extension).
We stop extending it in each direction whenever we hit a vertex of degree at least three;
if it is the same vertex in both directions then we extend it only in one direction
(since a path cannot contain the same vertex more than once).

The second phase begins when all vertices of degree two are already covered.
In the second phase,
ideally we would find a matching between those uncovered vertices of degree at least three.
To get a $2$-approximation
we arbitrarily select a vertex of degree at least three,
view it as a path of length one,
and arbitrarily extend it until it is maximal.
This finishes the description of the linear-time algorithm.

For correctness of the first phase,
the crucial observation is that each vertex of degree two has two be covered by at least one path.
For the second phase,
$2$-approximation follows since
each maximal path can cover at most two vertices of degree at least three.
\end{proof}
}

Now we are ready to design a polynomial linear-time parameterized algorithm for \hyp with respect to
the minimum maximal paths cover number.

\begin{theorem}%
  \label{thm:maximal-paths}
	Let~$G = (V,E)$ be a graph and $k$ be its minimum maximal paths cover number.
	Then, \hyp can be solved in~$O(k^4 (n+m))$ time.
\end{theorem}
{
\begin{proof}
We begin with some preprocessing.
First,
we apply \cref{rrule:degree1} to have a graph with no vertices of degree one.
Second,
we check whether there are any isolated cycles;
if there are,
then we consider the largest isolated cycle,
and compute its hyperbolicity.
If its hyperbolicity is at least $\delta$ then we have a yes-instance and we halt;
otherwise,
we remove all isolated cycles and continue.

Now we use \cref{lem:maximal-paths} to get a set of at most $2k$ maximal paths which cover $G$.
By initiating a breadth-first search from each of the endpoints of those maximal paths,
we can compute the pairwise distances between those endpoints in $O(k (n+m))$ time.
Thus,
for the rest of the algorithm we assume that we can access the distances between any two vertices
which are endpoints of those maximal paths in constant time.

	Let~$(a,b,c,d)$ be a quadruple such that~$\delta(a,b,c,d)=\delta(G)$. Since the set~${\cal P}$ covers all vertices of~$G$, each vertex of~$a$, $b$, $c$, and~$d$ belongs to some path~$P\in {\cal P}$. Since~$|\mathcal{P}| = k$, there are~$O(k^4)$ possibilities to assign the vertices~$a$, $b$, $c$, and~$d$ to paths in~$\mathcal{P}$. For each possibility we compute the maximum hyperbolicity respecting the assignment in linear time, that is, we compute the positions of the vertices on their respective paths that maximize~$\delta(a,b,c,d)$. We achieve the running time by formulating an integer linear program (ILP) with a constant number of variables and constraints whose coeffecients have value at most~$n$. 
	 
	To this end, denote with~$P_a, P_b, P_c, P_d \in \mathcal{P}$ the paths containing~$a,b,c,d$, respectively.
	We assume for now that these paths are different and deal later with the case that one path contains at least two vertices from~$a,b,c,d$.
	Let~$a_1$ and~$a_2$ ($b_1,b_2,c_1,c_2,d_1,d_2$) be the endpoints of~$P_a$ ($P_b, P_c, P_d$, respectively).
	Furthermore, denote by~$\ell(P)$ the length of a path~$P \in \mathcal{P}$, that is, the number of its edges.\todo{C: I propose to change to~$|P|$ because this will make the proof of Lemma 7 more readable.}
	Without loss of generality assume that~$D_1 \le D_2 \le D_3$.
	We now compute the positions of the vertices on their respective paths that maximize~$D_1 - D_2$ by solving an ILP.
	Recall that~$\distP{v_1}{v}{P_v}$ denotes the distance of~$v$ to~$v_1$ on~$P_v$.
	Thus, $\distP{v_1}{v}{P_v} + \distP{v}{v_2}{P_v} = \ell(P_v)$ and~$\distP{v_1}{v}{P_v} \ge 0$ and~$\distP{v}{v_2}{P_v} \ge 0$.
	The following is a compressed description of the ILP containing the minimum function. 
	We describe below how to remove it.
	\begin{align}
		\text{maximize:} 				&& & D_1 - D_2 \\
		\text{subject to:} 				&& D_1 & = \dist{a}{b}+\dist{c}{d} \\
										&& D_2 & = \dist{a}{c}+\dist{b}{d} \\
										&& D_3 & = \dist{a}{d}+\dist{b}{c} \\
										&& D_1 &\le D_2 \le D_3 \\
		\forall x \in \{a,b,c,d\}:		&& \ell(P_x) & = \distP{x_1}{x}{P_x} + \distP{x}{x_2}{P_x} \label{line:vars} \\ 
		\forall x,y \in \{a,b,c,d\}:	&& \dist{x}{y} & = \min \left\{
		\begin{array}{c}
			\distP{x_1}{x}{P_x} + \dist{x_1}{y_1} + \distP{y_1}{y}{P_y}, \\
			\distP{x_1}{x}{P_x} + \dist{x_1}{y_2} + \distP{y}{y_2}{P_y}, \\
			\distP{x}{x_2}{P_x} + \dist{x_2}{y_1} + \distP{y_1}{y}{P_y}, \\
			\distP{x}{x_2}{P_x} + \dist{x_2}{y_2} + \distP{y}{y_2}{P_y}  
		\end{array} \right\}\label{line:min}
		\end{align}
	First, observe that the ILP obviously has a constant number of variables. The only constant coefficients are~$\dist{x_i}{y_j}$ for~$x,y \in \{a,b,c,d\}$ and~$i,j \in \{1,2\}$ and obviously have value at most~$n-1$.
	To remove the minimization function in \cref{line:min}, we use another case distinction:
	We simply try all possibilities of which value is the smallest one and adjust the ILP accordingly. 
	For example, for the case that the minimum in \cref{line:min} is~$\distP{x}{x_1}{P_x} + \dist{x_1}{y_1} + \distP{y_1}{y}{P_y}$, we replace this equation by the following:
	\begin{align*}
		\dist{x}{y} =   \distP{x_1}{x}{P_x} + \dist{x_1}{y_1} + \distP{y_1}{y}{P_y} \\
		\dist{x}{y} \le \distP{x_1}{x}{P_x} + \dist{x_1}{y_2} + \distP{y}{y_2}{P_y} \\
		\dist{x}{y} \le \distP{x}{x_2}{P_x} + \dist{x_2}{y_1} + \distP{y_1}{y}{P_y} \\
		\dist{x}{y} \le \distP{x}{x_2}{P_x} + \dist{x_2}{y_2} + \distP{y}{y_2}{P_y} 
	\end{align*}
	There are four possibilities of which value is the smallest one, and we have to
        consider each of them independently for each of the $\binom{4}{2} = 6$ pairs.
        Hence, for each assignment of the vertices~$a$, $b$, $c$, and~$d$ to paths
        in~$\mathcal{P}$, we need to solve $4 \cdot 6 = 24$ different ILPs in order to
        remove the minimization function.%
        Since each ILP has a constant number of variables and constraints, this takes~$L^{O(1)}$ time where~$L=O(\log n)$ is the total size of the ILP instance (for example by using the algorithm of Lenstra~\cite{Len83}).
	
	It remains to discuss the case that at least two vertices of~$a$, $b$, $c$, and~$d$ are assigned to the same path $P \in \mathcal{P}$.
	We show the changes in case that~$a$, $b$, and $c$ are mapped to~$P_{a} \in \mathcal{P}$.
	We assume without loss of generality that the vertices~$a_1,a,b,c,a_2$ appear in this order in~$P$ (allowing $a=a_1$ and~$c=a_2$). 
	The adjustments for the other cases can be done in a similar fashion.
	The objective function as well as the first four lines of the ILP remain unchanged.
	\cref{line:vars} is replaced with the following:
	\begin{align*}
		\ell(P_a) & = \distP{a_1}{a}{P_a} + \distP{a}{b}{P_a} + \distP{b}{c}{P_a} + \distP{c}{a_2}{P_a} \\
		\ell(P_d) & = \distP{d_1}{d}{P_d} + \distP{d}{d_2}{P_a}
	\end{align*}
	To ensure that \cref{line:min} works as before, we add the following:
	\begin{align*}
		\distP{a}{a_2}{P_a} & = \distP{a}{b}{P_a} + \distP{b}{c}{P_a} + \distP{c}{a_2}{P_a} \\
		\distP{b_1}{b}{P_b} & = \distP{a_1}{a}{P_a} + \distP{a}{b}{P_a} \\
		\distP{b}{b_2}{P_b} & = \distP{b}{c}{P_a} + \distP{c}{a_2}{P_a} \\
		\distP{c_1}{c}{P_c} & = \distP{a_1}{a}{P_a} + \distP{a}{b}{P_a} + \distP{b}{c}{P_a} \\
		\distP{c}{c_2}{P_c} & = \distP{c}{a_2}{P_a}  	
	\end{align*}
\end{proof}
}

\paragraph{Feedback edge number.}

We next show a polynomial linear-time parameterized algorithm with respect to the parameter feedback edge number~$k$.
The idea is to show that a graph that is reduced with respect to~\cref{rrule:degree1} contains $O(k)$ maximal paths.

\begin{theorem}%
\label{thm:fes}
	\hyp can be computed in $O(k^4(n+m))$ time, where~$k$ is the feedback edge number.
\end{theorem}
{
\begin{proof}
	The first step of the algorithm is to reduce the input graph exhaustively with respect to \cref{rrule:degree1}.
	By \cref{lem:rrule-1-correct-lin-time} we can exhaustively apply \cref{rrule:degree1} in linear time.
	
	Denote by~$X \subseteq E$ a minimum feedback edge set for the reduced graph~$G = (V,E)$ and observe that $|X| = k$. We will show that the minimum maximal paths cover number of~$G$ is~$O(k)$. More precisely, we show the slightly stronger claim that the number of maximal paths in~$G$ is~$O(k)$. 

Observe that all vertices in~$G$ have degree at least two since~$G$ is reduced with respect to \cref{rrule:degree1}.    
	Thus, every leaf of~$G-X$ is incident with at least one feedback edge which implies that there are at most~$2k$ leaves in~$G-X$. 
	Moreover, since~$G-X$ is a forest, the number of vertices with degree at least three in~$G-X$ is at most the number of leaves in~$G-X$ and thus at most~$2k$. This implies that the number of maximal paths in~$G-X$ is at most~$2k$ (each maximal path corresponds to an edge in the forest obtained from~$G-X$ by contracting all degree-two vertices).
        
        We now show the bound for~$G$ by showing that an insertion of an edge into any graph~$H$ increases the number of maximal paths by at most five. Hence, consider a graph~$H$ and let~$\{u,v\}$ be an edge that is inserted into~$H$; denote the resulting graph by~$H'$. First, each edge can be part of at most one maximal path in any graph. Therefore, there is at most one maximal path~$P$ in~$H'$ that contains~$\{u,v\}$. The only vertices of~$P$ that can be in further in maximal paths of~$H'$ are the endpoints of~$P$. If an endpoint~$w$ of~$P$ has degree at least three, in~$H$ then each maximal path of~$H$ containing this endpoint is also maximal path in~$H'$. Otherwise, that is, if $w$ has degree two in~$H$, then there can be at most two new maximal paths containing~$w$, one for each edge that is incident with~$w$ in~$H$. Thus, the number of maximal paths containing~$w$ and different from~$P$ increases by at most two. Therefore, the insertion of the~$k$ edges of~$X$ in~$G-X$ increases the number of maximal paths by at most~$5k$. 
	Thus~$G$ contains at most~$7k$ maximal paths. The statement
        of the theorem now follows from \cref{thm:maximal-paths}.
\end{proof}
}

\paragraph{Number of vertices with degree at least three.}
We finally show a polynomial-linear time parameterized algorithm with respect to the number~$k$ of vertices with degree three or more.
To this end, we use the following data reduction rule to bound the number of maximal paths in the graph by~$O(k^2)$ (in order to make use of~\cref{thm:maximal-paths}).

\begin{rrule}\label{rrule:paths-between-high-degree-vertices}
	Let~$G=(V,E)$ be a graph, $u,v \in V^{\ge 3}_{G}$ be two vertices of degree at least three, and~$\mathcal{P}_{uv}$ be the set of maximal paths in~$G$ with endpoints~$u$ and~$v$.
	Let~$\mathcal{P}_{uv}^9 \subseteq \mathcal{P}_{uv}$ be the set containing the shortest path, the four longest even-length paths, and the four longest odd-length paths in~$\mathcal{P}_{uv}$.
	If~$\mathcal{P}_{uv} \setminus \mathcal{P}_{uv}^9 \neq \emptyset$, then delete in~$G$ all inner vertices of the paths in~$\mathcal{P}_{uv} \setminus \mathcal{P}_{uv}^9$.
\end{rrule}

\begin{lemma}%
  \label{lem:paths-rule-correct}
	\cref{rrule:paths-between-high-degree-vertices} is correct and can be exhaustively applied in linear time.
\end{lemma}

{
\begin{proof}
	We first prove the running time. 
	We compute in linear time the set~$V^{\ge 3}_{G}$ of all vertices with degree at least three.
	Then for each~$v \in V^{\ge 3}_{G}$ we do the following.
	Starting from~$v$, we perform a modified breadth-first search that stops at vertices in~$V^{\ge 3}_{G}$. 
	Let~$R(V^{\ge 3}_{G},v)$ denote the visited vertices and edges.
	Observe that~$R(V^{\ge 3}_{G},v)$ consists of~$v$, some degree-two vertices, and all vertices of~$V^{\ge 3}_{G}$ that can be reached from~$v$ via maximal paths in~$G$. 
	Furthermore, with the breadth-first search approach we can also compute for all~$u \in R(V^{\ge 3}_{G},v) \cap V^{\ge 3}_{G}$ with~$u \neq v$ the number of maximal paths between~$u$ and~$v$ and their respective lengths.
	Then, in time linear in~$|R(V^{\ge 3}_{G},v)|$, we remove the paths in~$\mathcal{P}_{uv} \setminus \mathcal{P}_{uv}^9$ for all~$u \in R(V^{\ge 3}_{G},v) \cap V^{\ge 3}_{G}$.
	Thus, we can apply \cref{rrule:paths-between-high-degree-vertices} for each~$v \in V^{\ge 3}_{G}$ in~$O(|R(V^{\ge 3}_{G},v)|)$ time.
	Altogether, the running time is
	$$O(\sum_{v\in V^{\ge 3}_{G}} |R(V^{\ge 3}_{G},v)|) = O(n+m)$$
	where the equality follows from the fact each edge and each maximal path in~$G$ is visited twice by the modified breadth-first search.
	
	We now prove the correctness of the data reduction rule.
	To this end, let~$G = (V,E)$ be the input graph, let~$P \in \mathcal{P}_{uv} \setminus \mathcal{P}_{uv}^9$ be a maximal path from~$u$ to~$v$ whose inner vertices are removed by the application of the data reduction rule, and let~$G' = (V', E')$ be the resulting graph.
	We show that~$\delta(G) = \delta(G')$.
	The correctness of \cref{rrule:paths-between-high-degree-vertices} follows then from iteratively applying this argument.
	First, observe that since~$\mathcal{P}_{uv}^9$ contains the shortest maximal path of~$\mathcal{P}_{uv}$, it follows that~$u$ and~$v$ have the same distance in~$G$ and~$G'$.
	Furthermore, it is easy to see that each pair of vertices~$w,w' \in V'$ has the same distance in~$G$ and~$G'$ (\cref{rrule:paths-between-high-degree-vertices} removes only paths and does not introduce degree-one vertices).
	Hence, we have that~$\delta(G) \ge \delta(G')$ and it remains to show that~$\delta(G) \le \delta(G')$
	
	Towards showing that~$\delta(G) \le \delta(G')$, let~$a,b,c,d \in V$ be the four vertices defining the hyperbolicity of~$G$, that is, $\delta(G) = \delta(a,b,c,d)$.
	If~$P$ does not contain any of these four vertices, then we are done.
	Thus, assume that~$P$ contains at least one vertex from~$\{a,b,c,d\}$. 
	(For convenience, we say in this proof that a path~$Q$ contains a vertex~$v$ if~$v$ is an inner vertex of~$Q$ because \cref{rrule:paths-between-high-degree-vertices} does neither delete~$u$ nor~$v$.) 
	We next make a case distinction on the number of vertices of~$\{a,b,c,d\}$ that are contained in~$P$.

	\emph{Case (I): $P$ contains one vertex of~$\{a,b,c,d\}$.}
	Without loss of generality assume~$P$ contains~$a$. We show that we can replace~$a$ by another vertex~$a'$ in a path~$P' \in \mathcal{P}_{uv}^9$ such that~$\delta(a,b,c,d) = \delta(a',b,c,d)$.
	Since~$P$ contains~$a$, we can chose~$P'$ as one of the four (odd/even)-length longest paths in $\mathcal{P}_{uv}^9$ such that 
	\begin{itemize}
		\item $\ell(P') - \ell(P)$ is nonnegative and even (either both lengths are even or both are odd) and 
		\item $P'$ contains no vertex of~$\{b,c,d\}$. 
	\end{itemize}
	Since~$P$ is removed by \cref{rrule:paths-between-high-degree-vertices}, it follows that~$\ell(P) \le \ell(P')$.
	We chose~$a'$ on~$P'$ such that~$\distP{u}{a'}{P'} = \distP{u}{a}{P} + (\ell(P') - \ell(P))/2$.
	Observe that this implies that~$\distP{a'}{v}{P'} = \distP{a}{v}{P} + (\ell(P') - \ell(P))/2$ and thus
	$$ \distP{u}{a}{P} - \distP{a}{v}{P} = \distP{u}{a'}{P'} - \distP{a'}{v}{P'}.$$
	Recall that
	\begin{align*}
		D_1 :=\dist{a}{b}+\dist{c}{d}, && D_2 :=\dist{a}{c}+\dist{b}{d}, \text{ and } && D_3 :=\dist{a}{d}+\dist{b}{c}.
	\end{align*}
	Denote with~$D'_1$, $D'_2$, and~$D'_3$ the respective distance sums resulting from replacing~$a$ with~$a'$, for example $D'_1 = \dist{a'}{b} + \dist{c}{d}$.
	Observe that by the choice of~$a'$ we increased all distance sums by the same amount, that is, for all~$i\in\{1,2,3\}$ we have~$D'_i = D_i + (\ell(P') - \ell(P))/2$. 
	Since~$\delta(a,b,c,d) = D_i - D_j$ for some~$i,j \in \{1,2,3\}$, we have that
	$$\delta(G') = \delta(a',b,c,d) = D'_i - D'_j = \delta(a,b,c,d) = \delta(G).$$
	
	\emph{Case (II): $P$ contains two vertices of~$\{a,b,c,d\}$.}
	Without loss of generality, assume that~$P$ contains~$a$ and~$b$ but not~$c$ and~$d$.
	We follow a similar pattern as in the previous case and again use the same notation.
	Let~$P', P'' \in \mathcal{P}_{uv}^9$ be the two longest paths such that both~$P'$ and~$P''$ do neither contain~$c$ nor~$d$ and both~$\ell(P') - \ell(P)$ and~$\ell(P'') - \ell(P)$ are even.
	We distinguish two subcases:
	
        \emph{Case (II-1):~$D_1$ is not the largest sum ($D_1 < D_2$ or~$D_1 < D_3$).}  We replace~$a$ and~$b$ with~$a'$ and~$b'$ on~$P'$ such that~$\distP{u}{a'}{P'} = \distP{u}{a}{P} + (\ell(P') - \ell(P))/2$ and~$\distP{u}{b'}{P'} = \distP{u}{b}{P} + (\ell(P') - \ell(P))/2$.
	Thus, $D'_1 = D_1$ since~$\dist{a}{b} = \dist{a'}{b'}$.
	However, for~$i \in \{2,3\}$ we have~$D'_i = D_i + (\ell(P') - \ell(P))/2$. 
	Since either~$D_2$ or~$D_3$ was the largest distance sum, we obtain
	$$ \delta(G) = \delta(a,b,c,d) = D_i - D_j \le D'_i - D'_{j'} = \delta(a',b',c,d) = \delta(G')$$
	for some~$i \in \{2,3\}$, $j,j' \in \{1,2,3\}$, $i \neq j$, and~$i \neq j'$.
	
		\emph{Case (II-2):~$D_1$ is the largest sum ($D_1 \ge D_2$ and~$D_1 \ge D_3$).} We need another replacement strategy since we did not increase~$D_1$ in case (II-1).
	In fact, we replace~$a$ and~$b$ with two vertices on different paths~$P'$ and~$P''$.
	We replace~$a$ with~$a'$ on~$P'$ and~$b$ with~$b'$ on~$P''$ such that~$\distP{u}{a'}{P'} = \distP{u}{a}{P} - (\ell(P') - \ell(P))/2$ and~$\distP{u}{b'}{P''} = \distP{u}{b}{P} - (\ell(P'') - \ell(P))/2$.
	Observe that for~$i \in \{2,3\}$ it holds that
	$$D'_i = D_i + (\ell(P') - \ell(P))/2 + (\ell(P'') - \ell(P))/2.$$
	Moreover, since~$a'$ and~$b'$ are on different maximal paths, we also have
	\begin{align*}
		\dist{a}{b} & \le \min_{x \in \{u,v\}}\{\distP{x}{a}{P}+\distP{x}{b}{P}\} \\
					& = \min_{x \in \{u,v\}}\{\distP{x}{a'}{P'}+\distP{x}{b'}{P''}\} - \frac{\ell(P') - \ell(P)}{2} - \frac{\ell(P'') - \ell(P)}{2} = \dist{a'}{b'} 
	\end{align*}
	and thus~$D'_1 \ge D_1 + (\ell(P') - \ell(P))/2 + (\ell(P'') - \ell(P))/2$.
	Hence, we have
	$$ \delta(G) = \delta(a,b,c,d) = D_1 - D_j \le D'_1 - D'_{j} = \delta(a,b,c,d) = \delta(G') $$
	for some~$j \in \{2,3\}$. 
 	
	\emph{Case (III): $P$ contains all four vertices of~$\{a,b,c,d\}$.} 
	We consider two subcases.
	
        \emph{Case (III-1): the union of the shortest paths between these four vertices  induces a path.} In this case, we have~$\delta(G) = 0$ and thus trivially~$\delta(G) \le \delta(G')$.

        \emph{Case (III-2): the union of the shortest paths between these four vertices  induces a cycle.} 
               From \cref{lem:hyp-distance-bounded} we derive~$\delta(G) \le \ell(P)/2$ since at least two of the four vertices~$a,b,c,d$ have distance at most~$\ell(P)/4$. 
	We can replace the four vertices with four vertices on a path~$P' \in \mathcal{P}_{uv}^9$ such that~$\ell(P') - \ell(P)$ is nonnegative and even. Observe that if~$\ell(P') = \ell(P)$, then taking the vertices on the same positions as~$a,b,c,d$ gives a 4-tuple with the same distances. Hence, assume $\ell(P')>\ell(P)$.  
        Consider the union of the vertices on~$P'$ and the shortest path between~$u$ and~$v$. 
	The union of the shortest paths of all vertices in this set is a cycle of length at least~$\ell(P')+1\ge \ell(P)+2$. 
	By known results of \citet{KM02} there is a 4-tuple of cycle vertices~$A',b',c',d'$ such that~$\delta(a,b,c,d)\ge \lfloor (\ell(P')+1)/2\rfloor > \ell(P)/2$. Thus, we have
	$$\delta(G') \ge \delta(a',b',c',d') > \ell(P)/2 \ge \delta(G).$$
	
	\emph{Case (IV): $P$ contains three vertices of~$\{a,b,c,d\}$.}
	Without loss of generality, assume that~$P$ contains~$a$,~$b$, and~$c$ but not~$d$ and that~$a$ is the closest vertex to~$u$ on~$P$ and~$c$ is the closest vertex to~$v$ on~$P$ (that is, $a,b,c$  appear in this order on~$P$). 
	We distinguish two subcases.
	
	\emph{Case (IV-1): $\distP{a}{c}{P} = \dist{a}{c}$.} We follow a similar pattern as in case (I) and use the same notation. 
	Again, there is a $P' \in \mathcal{P}_{uv}^9$ such that $\ell(P') - \ell(P)$ is even (either both lengths are even or both are odd) and $P'$ does not contain~$d$.
	We replace each vertex~$a,b,c$ as in case (I), that is, for each~$x \in \{a,b,c\}$ we chose~$x'$ on~$P'$ such that~$\distP{u}{x'}{P'} = \distP{u}{x}{P} + (\ell(P') - \ell(P))/2$.
	Observe that only the distances between~$d$ and the other three vertices change.
	Thus, we have again for all~$i\in\{1,2,3\}$ that~$D'_i = D_i + (\ell(P') - \ell(P))/2$ and hence~$\delta(G) = \delta(G')$.

	\emph{Case (IV-2):~$\distP{a}{c}{P} > \dist{a}{c}$.}  We use again a similar strategy as in case (I) and use the same notation. 
	Again, there is a $P' \in \mathcal{P}_{uv}^9$ such that $\ell(P') - \ell(P)$ is even (either both lengths are even or both are odd) and $P'$ does not contain~$d$.
	We replace the vertices~$a,b,c$ with~$a',b',c'$ on~$P'$ such that
	\begin{itemize}
		\item $\dist{a}{u} = \distP{a}{u}{P} = \distP{a'}{u}{P'} = \dist{a'}{u}$,
		\item $\dist{c}{v} = \distP{c}{v}{P} = \distP{c'}{v}{P'} = \dist{c'}{v}$, 
		\item $\distP{b}{u}{P} = \distP{b'}{u}{P'} - (\ell(P') - \ell(P))/2$, and
		\item $\distP{b}{v}{P} = \distP{b'}{v}{P'} - (\ell(P') - \ell(P))/2$.
	\end{itemize}
	Note that since~$\distP{a}{c}{P} > \dist{a}{c}$, it follows that the distances not involving~$b$ remain unchanged, that is, $\dist{a}{b} = \dist{a'}{b'}$, and for~$x \in \{a,b\}$ we have~$\dist{x}{d} = \dist{x'}{d}$.
	Furthermore, all distances involving~$b$ increase by~$(\ell(P') - \ell(P))/2$, that is, $\dist{b}{d} = \dist{b'}{d} - (\ell(P') - \ell(P))/2$ and for~$x \in \{a,b\}$ we have~$\dist{b}{x} = \dist{b'}{x'}$.
	Thus, we have again for all~$i\in\{1,2,3\}$ that~$D'_i = D_i + (\ell(P') - \ell(P))/2$ and hence~$\delta(G) = \delta(G')$.
\end{proof}
}
Observe that if the graph~$G$ is reduced with respect to \cref{rrule:paths-between-high-degree-vertices}, then there exist for each pair~$u,v \in V^{\ge 3}_G$ at most nine maximal paths with endpoints~$u$ and~$v$.
Thus, $G$ contains at most~$O(k^2)$ maximal paths and using \cref{thm:maximal-paths} we arrive at the following.

\begin{theorem}
\label{thm:degree3vertices}
	\hyp can be solved in~$O(k^8 (n+m))$ time, where~$k$ is the number of vertices with degree at least three.
\end{theorem}

\section{Parameter Vertex Cover}

A \emph{vertex cover} of a graph~$G = (V, E)$ is a subset~$W \subseteq V$ of vertices of~$G$ such that each edge in~$G$ is incident to at least one vertex in~$W$.
Deciding whether a graph~$G$ has a vertex cover of size at most~$k$ is \NP-complete in general~\cite{GJ79}.
There is, however, a simple linear-time factor-$2$ approximation (see, e.g., \cite{PapadimitriouS82}). %
In this section, we consider the size~$k$ of a vertex cover as the parameter.
We show 
that we can solve \hyp in time linear in~$|G|$, but exponential in~$k$; further, we show 
that, unless SETH fails, we cannot do asymptotically better.

\paragraph{A Linear-Time Algorithm Parameterized by the Vertex Cover Number.}\label{section:vcub}

We prove that \hyp can be solved in time linear in the size of the graph and exponential in the size~$k$ of a vertex cover.
This result is based on a linear-time computable kernel of size~$O(2^k)$, that can be obtained by exhaustively applying the following reduction rule.

\begin{rrule}\label{rrule:twins}
	If there are at least five vertices~$v_1, v_2, \ldots, v_\ell \in V$, $\ell > 4$, with the same (open) neighborhood~$N(v_1) = N(v_2) = \ldots = N(v_\ell)$, then delete~$v_5, \ldots, v_\ell$.
\end{rrule}
We next show that the above rule is correct, can be applied in linear time, and leads to a kernel for the parameter vertex cover number.
\begin{lemma}%
  \label{lemma:vcredrules}
	\cref{rrule:twins} is correct and can be applied exhaustively in linear time.
	Furthermore, if \cref{rrule:twins} is not applicable, then the graph contains at most $k + 4 \cdot 2^k$ vertices and~$O(k \cdot 2^k)$ edges, where~$k$ is the vertex cover number.
\end{lemma}

{
  \begin{proof}
	  Let $G=(V,E)$ be the input graph with a vertex cover~$W \subseteq V$ of size~$k$ and let~$v_1, v_2, \ldots, v_\ell \in V$, $\ell > 4$, be vertices with the same open neighborhood.

	  First, we show that \cref{rrule:twins} is correct, that is, $\delta(G[V \setminus \{v_5, \ldots, v_\ell\}])=\delta(G)$.
	  To see this, consider two vertices $v_i$, $v_j$ with the same open neighborhood, and consider any other vertex $u$.
	  The crucial observation is that $\dist{u}{v_i} = \dist{u}{v_j}$.
	  This means that the two vertices are interchangeable with respect to the hyperbolicity. 
	  In particular, if $v_i,v_j \in V$ have the same open neighborhood, then $\delta(v_i,x,y,z)=\delta(v_j,x,y,z)$ for every $x,y,z\in V\setminus \{v_i,v_j\}$.
	  As the hyperbolicity is obtained from a quadruple, it is sufficient to consider at most four vertices with the same open neighborhood.
	  We conclude that $\delta(G[V \setminus \{v_5, \ldots, v_\ell\}])=\delta(G)$. 
	  
	  Next we show how to exhaustively apply \cref{rrule:twins} in linear time.
	  To this end, we apply in linear time a \emph{partition refinement}~\cite{HabibP10} to compute a partition of the vertices into twin classes. 
	  Then, for each twin class we remove all but 4 (arbitrary) vertices.
	  Overall, this can be done in linear time.

	  Since $|W|\leq k$, it follows that there are at most $2^{k}$ pairwise-different neighborhoods (and thus twin classes) in~$V \setminus W$.
	  Thus, if \cref{rrule:twins} is not applicable, then the graph consists of the vertex cover~$W$ of size~$k$ plus at most~$4 \cdot 2^{k}$ vertices in~$V \setminus W$.
	  Furthermore, since~$W$ is a vertex cover, it follows that the graph contains at most~$ 4k \cdot 2^k$ edges. 
  \end{proof}
}

With \cref{rrule:degree1} we can compute in linear time an equivalent instance having a bounded number of vertices. 
Applying on this instance the trivial~$O(n^4)$-time algorithm yields the following.

\begin{theorem}
\label{theorem:vc}
	\hyp can be computed in $O(2^{4k} + n + m)$ time, where~$k$ denotes the size of a vertex cover of the input graph. 
\end{theorem}

\paragraph{SETH-based Lower bounds.}\label{ssec:vclowerbound}

We show that, unless SETH breaks, the~$2^{O(k)} + O(n+m)$-time algorithm obtained in the previous subsection cannot be improved to an algorithm even with running time $2^{o(k)}\cdot (n^{2-\epsilon})$.
This also implies, that, assuming SETH, there is no kernel with~$2^{o(k)}$ vertices computable in~$O(n^{2-\epsilon})$ time, i.\,e.\ the kernel obtained by applying \cref{rrule:twins} cannot be improved significantly.
The proof follows by a reduction from the following problem.

\problemdef
  {Orthogonal Vectors}
  {Two sets~$\vec{A}$ and~$\vec{B}$ each containing~$n$ binary vectors of length~$\ell=O(\log n)$.}
  {Are there two vectors~$\vec{a}\in \vec{A}$ and~$\vec{b}\in \vec{B}$ such that~$\vec{a}$ and~$\vec{b}$ are orthogonal,
  that is,
  such that there is no position~$i$ for which~$\vec{a}[i]=\vec{b}[i]=1$?}  

Williams and Yu~\cite{WY14} proved that,
if \textsc{Orthogonal Vectors} can be solved in~$O(n^{2-\epsilon})$ time,
then SETH breaks. 
We provide a linear-time reduction from~\textsc{Orthogonal Vectors} to \hyp
where the graph~$G$ constructed in the reduction contains~$O(n)$ vertices and admits a vertex cover of size~$O(\log(n))$
(and thus contains~$O(n\cdot \log n)$ edges). 
The reduction then implies that,
unless SETH breaks,
there is no algorithm solving \hyp in time polynomial in the size of the vertex cover and linear in the size of the graph.
We mention that Borassi et al.~\cite{BCH16} showed that under the SETH \hyp cannot be solved in~$O(n^{2-\epsilon})$.
However, the instances constructed in their reduction have a minimum vertex cover of size~$\Omega(n)$.
Note that our reduction is based on ideas from the reduction of Abboud et al.~\cite{AWW16} for the \textsc{Diameter} problem.

\begin{theorem}%
  \label{thm:SETH-lowerbound}
  Assuming SETH, \hyp cannot be solved in~$2^{o(k)}\cdot (n^{2-\epsilon})$ time, even on graphs with~$O(n \log n)$ edges, diameter four, and domination number three. 
  Here,~$k$ denotes the vertex cover number of the input graph.
\end{theorem}

{
  \begin{proof}
    We reduce any instance~$(\vec{A},\vec{B})$ of \textsc{Orthogonal Vectors} to an instance~$(G,\delta)$ of \hyp,
    where we construct the graph~$G$ as follows
    (we refer to \cref{fig:orthogvec} for a sketch of the construction). 
    
    Make each~$\vec{a}\in \vec{A}$ a vertex~$a$ and each~$\vec{b}\in \vec{B}$ a vertex~$b$ of~$G$,
    and denote these vertex sets by~$A$ and~$B$, respectively. 
    Add two vertices for each of the $\ell$~dimensions, that is, add the vertex set~$C:=\{c_1,\ldots, c_\ell\}$ and the vertex set~$D=\{d_1,\ldots, d_\ell\}$ to~$G$ and make each of~$C$ and~$D$ a clique. 
    Next, connect each~$a\in A$ to the vertices of~$C$ in the natural way, that is, add an edge between~$a$ and~$c_i$ if and only if~$\vec{a}[i]=1$. 
    Similarly, add an edge between~$b\in B$ and~$d_i\in D$ if and only if~$\vec{b}[i]=1$. 
    Moreover, add the edge set~$\{\{c_i,d_i\}\mid i\in [\ell]\}$. 
    This part will constitute the central gadget of our construction. 
    
    Our aim is to ensure that the maximum hyperbolicity is reached for 4-tuples~$(a,b,c,d)$ such that~$a\in A$,~$b\in B$, and~$a$ and~$b$ are orthogonal vectors. 
    The construction of~$G$ is completed by adding two paths~$(u_A,u,u_B)$ and~$(v_A,v,v_B)$, and making~$u_A$ and~$v_A$ adjacent to all vertices in~$A\cup C$ and~$u_B$ and~$v_B$ adjacent to all vertices in~$B\cup D$.
    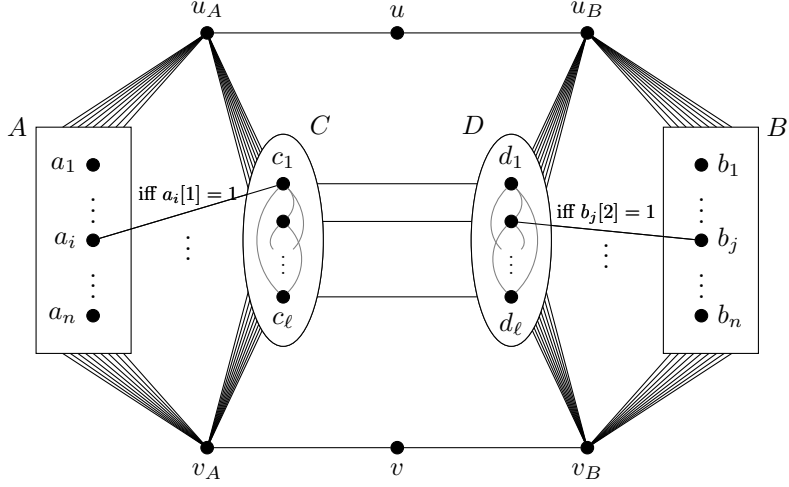
\begin{figure}[t]
    \centering
    \begin{tikzpicture}

      \tikzstyle{node}=[circle, fill, scale=1/2, draw]
      \def\ysh{0.25}
      \node (c1) at (-1.5,0-\ysh)[node]{};
      \node (c2) at (-1.5,-0.5-\ysh)[node]{};
      \node (ci) at (-1.5,-1-\ysh)[]{$\vdots$};
      \node (cell) at (-1.5,-1.5-\ysh)[node]{};

      \node (d1) at (1.5,0-\ysh)[node]{};
      \node (d2) at (1.5,-0.5-\ysh)[node]{};
      \node (ci) at (1.5,-1-\ysh)[]{$\vdots$};
      \node (dell) at (1.5,-1.5-\ysh)[node]{};

      \node (c) at (-1.5,-0.75-\ysh)[ellipse, minimum width=30pt, minimum height=80pt,draw,label=80:{$C$}]{};
      \node (d) at (1.5,-0.75-\ysh)[ellipse, minimum width=30pt, minimum height=80pt,draw,label=100:{$D$}]{};

      \draw (c1) -- (d1);
      \draw (c2) -- (d2);
      \draw (cell) -- (dell);

      \node (a1) at (-4,0)[node,label=180:{$a_1$}]{};
      \node (adots1) at (-4,-0.5)[]{$\vdots$};
      \node (ai) at (-4,-1)[node,label=180:{$a_i$}]{};
      \node (adots2) at (-4,-1.5)[]{$\vdots$};
      \node (an) at (-4,-2)[node,label=180:{$a_n$}]{};

      \node (b1) at (4,0)[node,label=0:{$b_1$}]{};
      \node (bdots1) at (4,-0.5)[]{$\vdots$};
      \node (bi) at (4,-1)[node,label=0:{$b_j$}]{};
      \node (bdots2) at (4,-1.5)[]{$\vdots$};
      \node (bn) at (4,-2)[node,label=0:{$b_n$}]{};

      \draw (-4.75,-2.5) rectangle (-3.5,0.5);
      \node at (-5,0.5)[]{$A$};
      \draw (4.75,-2.5) rectangle (3.5,0.5);
      \node at (5,0.5)[]{$B$};

      \node at (-2.75,-1)[]{$\vdots$};
      \node at (2.75,-1.1)[]{$\vdots$};

      \draw (ai) -- node[above,midway,scale=3/4]{iff $a_i[1]=1$}(c1);
      \draw (bi) -- node[above,midway,scale=3/4]{iff $b_j[2]=1$}(d2);

      \node (ua) at (-2.5,1.75)[node, label=90:{$u_A$}]{};
      \node (u) at (0,1.75)[node, label=90:{$u$}]{};
      \node (ub) at (2.5,1.75)[node, label=90:{$u_B$}]{};

      \node (va) at (-2.5,-3.75)[node, label=-90:{$v_A$}]{};
      \node (v) at (0,-3.75)[node, label=-90:{$v$}]{};
      \node (vb) at (2.5,-3.75)[node, label=-90:{$v_B$}]{};

      \draw (va) -- (v) -- (vb);
      \draw (ua) -- (u) -- (ub);

      \foreach \x in {1,2,...,9}{
	\draw (ua) --(-3.5-0.1*\x,0.5);
	\draw (ua) --([xshift=2*\x-10 pt] c.center);
      }

      \foreach \x in {1,2,...,9}{
	\draw (va) --(-3.5-0.1*\x,-2.5);
	\draw (va) -- ([xshift=2*\x-10 pt] c.center);
      }

      \foreach \x in {1,2,...,9}{
	\draw (ub) --(3.5+0.1*\x,0.5);
	\draw (ub) --([xshift=2*\x-10 pt] d.center);
      }

      \foreach \x in {1,2,...,9}{
	\draw (vb) --(3.5+0.1*\x,-2.5);
	\draw (vb) --([xshift=2*\x-10 pt] d.center);
      }

      \node (c) at (-1.5,-0.75-\ysh)[ellipse, minimum width=30pt, minimum height=80pt,draw,fill=white]{};
      \node (d) at (1.5,-0.75-\ysh)[ellipse, minimum width=30pt, minimum height=80pt,draw,fill=white]{};

      \node (c1) at (-1.5,0-\ysh)[node,label=90:{$c_1$}]{};
      \node (c2) at (-1.5,-0.5-\ysh)[node]{};
      \node (ci) at (-1.5,-1-\ysh)[scale=3/4]{$\vdots$};
      \node (cell) at (-1.5,-1.5-\ysh)[node,label=-90:{$c_\ell$}]{};
      \draw[-,gray] (c1) to [out=-45,in=45](c2);
      \draw[-,gray] (c1) to [out=-45,in=45](ci);
      \draw[-,gray] (c2) to [out=-45,in=45](cell);
      \draw[-,gray] (c2) to [out=-135,in=135](ci);
      \draw[-,gray] (c1) to [out=-135,in=135](cell);

      \node (d1) at (1.5,0-\ysh)[node,label=90:{$d_1$}]{};
      \node (d2) at (1.5,-0.5-\ysh)[node]{};
      \node (di) at (1.5,-1-\ysh)[scale=3/4]{$\vdots$};
      \node (dell) at (1.5,-1.5-\ysh)[node,label=-90:{$d_\ell$}]{};
      \draw[-,gray] (d1) to [out=-45,in=45](d2);
      \draw[-,gray] (d1) to [out=-135,in=135](di);
      \draw[-,gray] (d2) to [out=-135,in=135](dell);
      \draw[-,gray] (d2) to [out=-45,in=45](di);
      \draw[-,gray] (d1) to [out=-45,in=45](dell);
      
      \draw (ai) -- node[above,midway,scale=3/4]{iff $a_i[1]=1$}(c1);
      \draw (bi) -- node[above,midway,scale=3/4]{iff $b_j[2]=1$}(d2);

      \end{tikzpicture}
      \caption{Sketch of the construction described in the proof of \cref{thm:SETH-lowerbound}. 
      Ellipses indicate cliques, rectangles indicate independent sets. 
      Multiple edges to an object indicate that the corresponding vertex is incident to each vertex enclosed within that object.}
      \label{fig:orthogvec}
    \end{figure}

    Observe that $G$~contains $O(n)$~vertices, $O(n \cdot \log n)$ edges,
    and that the set~$V\setminus (A\cup B)$ forms a vertex cover in~$G$ of size~$O(\log n)$.
    Moreover, observe that~$G$ has diameter four.
    Note that each vertex in~$A\cup B\cup C\cup D$ is at distance two to each of~$u$ and~$v$.
    Moreover, $v_A$~and $v_B$~are at distance three to~$u$. 
    Analogously, $u_A$, $u_B$~are at distance three to~$v$.
    Furthermore~$u$ and~$v$ are at distance four.
    Finally, observe that $\{u_A,u_B,v\}$ forms a dominating set in~$G$.
    
    We complete the proof by showing that $(\vec{A},\vec{B})$~is a yes-instance of \textsc{Orthogonal Vectors} if and only if $G$~has hyperbolicity at least~$\delta=4$.

    \raproof{} Let~$(\vec{A},\vec{B})$ be a yes-instance, and let~$\vec{a}\in \vec{A}$ and~$\vec{b}\in \vec{B}$ be a pair of orthogonal vectors. 
    We claim that $\delta(a,b,u,v)=4$. 
    Since~$\vec{a}$ and~$\vec{b}$ are orthogonal, there is no $i\in[\ell]$ with $\vec{a}[i]=\vec{b}[i]=1$ and, hence, there is no path connecting $a$ and $b$ only containing two vertices in $C\cup D$, and it holds that~$\dist{a}{b}=4$. 
    Moreover, we know that $\dist{u}{v}=4$ as that~$\dist{a}{u}=\dist{b}{u}=\dist{a}{v}=\dist{a}{v}=2$. 
    Thus,~$\delta(a,b,u,v)=8-4=4$, and $G$ is 4-hyperbolic.

    \laproof{} Let~$S=\{a,b,c,d\}$ be a set of vertices such that~$\delta(a,b,c,d)\ge4$. 
    By \cref{lem:hyp-distance-bounded}, it follows that no two vertices of~$S$ are adjacent. 
    Hence, we assume without loss of generality that~$\dist{a}{b}=\dist{c}{d}=4$. 
    Observe that all vertices of~$C$ and~$D$ have distance at most three to all other vertices. 
    Similarly, each vertex of~$\{u_A,v_A,u_B,v_B\}$ has distance at most three to all other vertices.
    (Consider for example~$u_A$. 
    By construction, $u_A$ is a neighbor of all vertices in~$A\cup C\cup \{u\}$ and, hence, $u_A$~has distance at most two to~$v_A$ and to all vertices in~$D$. 
    Thus, $u_A$~has distance at most three to~$v$, $B$, $u_B$ and~$v_B$ and therefore to all vertices of~$G$. 
    The arguments for~$v_A$,~$u_B$, and~$v_B$ are symmetric).

    It follows that~$S\subseteq A\cup B\cup \{u,v\}$,
    and therefore at least two vertices in~$S$ are from~$A\cup B$. 
    Thus, assume without loss of generality that~$a$ is contained in~$A$. 
    By the previous assumption, we have that~$\dist{a}{b}=4$. 
    This implies that~$b\in B$ and~$\vec{a}$ and~$\vec{b}$ are orthogonal vectors, as every other vertex in~$V\setminus B$ is at distance three to~$a$ and each~$b'\in B$ with $\vec{b'}$ being non-orthogonal to~$\vec{a}$ is at distance three to~$a$. 
    Hence, $(\vec{A},\vec{B})$ is a yes-instance.
  \end{proof}
}

We remark that, with the above reduction, the hardness also holds for the variants in which we fix one vertex~($u$) or two vertices~($u$ and~$w$).
The reduction also shows that approximating the hyperbolicity of a graph within a factor of~$4/3-\epsilon$ cannot be done in strongly subquadratic time or with a PL-FPT running time. 

Next, we adapt the above reduction to obtain the following hardness result on graphs of bounded maximum degree.

\begin{theorem}%
  \label{thm:SETH-lowerbound2}
  Assuming SETH, \hyp cannot be solved in~$f(\Delta)\cdot (n^{2-\epsilon})$ time, where~$\Delta$ denotes the maximum degree of the input graph.
\end{theorem}

{
  \begin{proof}
    We reduce any instance~$(\vec{A},\vec{B})$ of \textsc{Orthogonal Vectors} to an instance~$(G,\delta)$ of \hyp as follows.
    
    We use the following notation. 
    For two sets of vertices~$X$ and~$Y$ with~$|X|=|Y|$, we say that we introduce \emph{matching paths} if we connect the vertices in~$X$ with the vertices in~$Y$ with
    paths with no inner vertices from~$X\cup Y$ such that for each $x\in X$, $x$ is connected to exactly one $y\in Y$ via one path and for each~$y\in Y$, $y$ is connected to exactly one $x\in X$ via one path.
    
    Let $G'$ be the graph obtained from the graph constructed in the proof of~\cref{thm:SETH-lowerbound} after deleting all edges.
    For each~$x_A$, $x\in \{u,v\}$, add two binary trees, $T_{x_A}^A$ with $n$~leaves and height at most~$\ceil{\log n}$, and $T_{x_A}^C$ with $\ell$~leaves and height at most~$\ceil{\log \ell}$.
    Connect each tree root by an edge with~$x_A$. 
    Next introduce matching paths between~$A$ and the leaves of~$T_{x_A}^A$ such that each shortest path connecting a vertex in~$A$ with $x_A$ is of length $h:=2(\ceil{\log(n)}+1)+1$.
    Similarly, introduce matching paths between~$C$ and the leaves of~$T_{x_A}^C$ such that each shortest path connecting a vertex in~$C$ with $x_A$ is of length $h$.
    Apply the same construction for $x_B$, $x\in \{u,v\}$, $B$, and $D$.
    
    For $x\in A\cup B$, we denote by~$|x|_1$ the number of 1's in the corresponding binary vector~$\vec{x}$.
    Moreover, for $c_i\in C$, we denote by~$|c_i|$ the number of vectors in~$A$ with a 1 as its $i$th entry.
    For $d_i\in D$, we denote by~$|d_i|$ the number of vectors in~$B$ with a 1 as its $i$th entry.
    
    For each vertex~$a\in A$, add a binary tree with $|a|_1$ leaves and height at most~$\ceil{\log |a|_1}$ and connect its root by an edge with $a$.
    For each~$i\in[\ell]$, add a binary tree with~$|c_i|$ leaves and height at most~$\ceil{\log |c_i|}$ and connect its root by an edge with~$c_i$.
    Next, construct matching paths between the leaves of all binary trees introduced for the vertices in~$A$ on the one hand, and the leaves of all binary trees introduced for the vertices in~$C$ on the other hand, such that the following holds:
    (i) for each~$a\in A$ and~$c_i\in C$, there is a path only containing the vertices of the corresponding binary trees if and only if $\vec{a}[i]=1$, and
    (ii) each of these paths is of length exactly~$h$.
    Apply the same construction for~$B$ and~$D$.
    
    Next, for each~$i\in[\ell]$, add a binary tree with~$\ell-1$ leaves and height at most~$\ceil{\log (\ell-1)}$ and connect its root by an edge with~$c_i$.
    Finally, add paths between the leaves of all binary trees introduced in this step such that
    (i) each leaf is incident to exactly one path,
    (ii) for each~$i,j\in[\ell]$, $i\neq j$, there is a path only containing the vertices of the corresponding binary trees, and
    (iii) each of these paths is of length exactly~$h$.
    Apply the same construction for~$D$.
    
    Finally, for each~$i\in[\ell]$, connect~$c_i$ with~$d_i$ via a path of length~$h$.
    Moreover, for~$x\in \{u,v\}$, connect~$x_A$ with~$x$ and~$x$ with~$x_B$ each via a path of length~$h$.
    This completes the construction of~$G$.
    Observe that the number of vertices in~$G$ is at most the number of vertices in the graph obtained from~$G'$ by replacing each edge with paths of length~$h$.
    As~$G'$ contains $O(n\log n)$ edges, the number of vertices in~$G$ is in~$O(n\log^2 n)$.
    Finally, observe that the vertices in~$C\cup D$ are the vertices of maximum degree which is five.
    
    Next, we discuss the distances of several vertices in the constructed graph.
    Observe that $u$ and $v$ are at distance $4h$.
    For~$x\in \{u,v\}$, the distance between~$x$ and~$x_A$ or~$x_B$ is~$h$, and the
    distance between~$x_A$ and~$x_B$ is~$2h$.
    The distance from any $c\in C$ to any $d\in D$ is at least $h$ and at most $2h$.
    Moreover, the distance between any $a\in A$ and $b\in B$ is at least $3h$ and at most~$4h$.
    
    \begin{claim}\label{claim:orthog4h}
    For any $a\in A$ and $b\in B$, $\dist{a}{b}=4h$ if and only if $\vec{a}$ and $\vec{b}$ are orthogonal.
    \end{claim}
    
    \begin{proof}[of \cref{claim:orthog4h}]
      \laproof{}
      Let $\vec{a}$ and $\vec{b}$ be orthogonal.
      Suppose that there is a shortest path~$P$ between $a$ and $b$ of length smaller than~$4h$.
      Observe that any shortest path between $a$ and $b$ containing $u$ or $v$ is of length~$4h$.
      Hence, $P$ contains vertices in~$C\cup D$.
      As the shortest paths from~$a$ to~$C$, $C$ to $B$, and $B$ to $b$ are each of length~$h$, the only shortest path containing vertices in $C\cup B$ of length smaller than $4h$ is of the form $(a,c_i,d_i,b)$ for some $c_i\in C$ and $d_i\in D$ (recall that the shortest path between any two vertices in~$C$ or~$D$ is of length~$h$). 
      Hence, $\vec{a}$ and~$\vec{b}$ have both a 1 as their $i$th entry, and thus are not orthogonal.
      This contradicts the fact that $\vec{a}$ and $\vec{b}$ form a solution.
      It follows that~$\dist{a}{b}=4h$.
      
      \raproof{}
      Let~$\vec{a}$ and~$\vec{b}$ be not orthogonal.
      Then there is an~$i\in[\ell]$ such that~$a[i]=b[i]=1$.
      Hence, there is a path~$(a,c_i,d_i,b)$ of length~$3h<4h$.
    \end{proof}
    
    Let~$M:=A\cup B\cup C\cup D\cup \{x,x_A,x_B\mid x\in\{u,v\}\}$.
    So far, we know that the only vertices that can be at distance~$4h$ are those in~$A\cup B\cup \{u,v\}$. 

    Consider any vertex~$p\in V(G)\setminus M$.
    Then~$p$ is contained in a shortest between two vertices~$x$ and~$y$ in~$M$ at distance~$h$.
    Moreover, $\max\{\dist{p}{x},\dist{p}{y}\}=:h'<h$.
    Let~$P_x^Y$ denote the set of inner vertices of the shortest path connecting~$x$ and~$x_A$, for $x\in\{u,v\}$, $Y\in\{A,B\}$.
    Moreover, let $M^*:=\{p\in P_x^Y\mid x\in\{u,v\}, Y\in\{A,B\}\}$.
    We first discuss the case where~$p\in M^*$.
    By symmetry, let $p\in P_u^A$. 
    Observe that for~$q\in P_v^B$ with $\dist{v}{q}=\dist{u}{p}$ holds $\dist{p}{q}=4h$. 
    
    Let $p\not\in M\cup M^*$.
    Then, we claim that for all vertices~$q\in V(G)$ it holds that~$\dist{p}{q}<4h$.
    Suppose not, so that there is some $q\in V(G)$ with $\dist{p}{q}\geq 4h$.
    Observe that~$q$ is not contained in a shortest path between~$x$ and~$y$.
    It follows that $\dist{x}{q}\geq 4h-h'>3h$ or $\dist{y}{q}\geq 4h-h'>3h$.
    Let $z\in\{x,y\}$ denote the vertex of minimal distance among the two, and let $\bar{z}$ denote the other one.
    Note that since $h$~is odd, the distances to~$z$ and~$\bar{z}$ are different.
    
    \emph{Case 1}: $q\in M$.
    Then $z,q\in A\cup B$, where $z$ and $q$ are not both contained in~$A$ or $B$.
    Recall that $p\not\in M\cup M^*$ and, hence, the case $z,q\in \{u,v\}$ is not possible.
    By symmetry, assume $z\in A$ and $q\in B$.
    As $\dist{z}{q}>3h$, it follows that $\bar{z}=c_i\in C$ for some~$i\in[\ell]$ with~$1=\vec{z}[i]\neq \vec{q}[i]$, or $\bar{z}\in\{u_A,v_A\}$.
    Hence, the distance of~$\bar{z}$ to~$q$ is at most the distance of~$z$ to~$q$, contradicting the choice of~$z$.
    
    \emph{Case 2}: $q\not \in M$.
    Then $q$ is contained in a shortest path between two vertices~$x',y'\in M$ of length~$h$.
    Moreover, $\max\{\dist{q}{x'},\dist{q}{y'}\}=:h''<h$.
    Consider a shortest path between $p$ and $q$ 
    and notice that it must contain $z$ and $z'\in\{x',y'\}$.
    It holds that~$\dist{z}{z'}\geq 4h-h'-h''>2h$.
    By symmetry, assume~$z\in A$, and~$z'\in D\cup\{u_B,v_B\}$ (recall that $p\not\in M\cup M^*$).
    Then $\bar{z}$ is in $C\cup \{u_A,v_A\}$, and hence of shorter distance to~$q$, contradicting the choice of~$z$.

    We proved that $\dist{p}{q}<4h$ for all $p\in V(G)\setminus (M\cup M^*)$, $q\in V(G)$.
    We conclude that the vertex set~$A\cup B\cup \{u,v\}\cup M^*$ is the only set containing vertices at distance~$4h$.
    Moreover, $G$~is of diameter~$4h$.
    
    We claim that $(\vec{A},\vec{B})$~is a yes-instance of \textsc{Orthogonal Vectors} if and only if $G$~has hyperbolicity at least~$\delta=4h$. 
    
    \raproof{}
    Let~$\vec{a}\in \vec{A}$ and~$\vec{b}\in \vec{B}$ be orthogonal.
    We claim that~$\delta(a,b,u,v)=4h$. 
    Observe that~$\dist{u}{v}=4h$, and that $\dist{a}{b}=4h$ by \cref{claim:orthog4h}.
    The remaining distances are~$2h$ by construction, and hence~$\delta(G)=\delta(a,b,u,v)=4h$.
    
    \laproof{}
    Let $\delta(G)=4h$ and let $w,x,y,z$ be a quadruple with $\delta(w,x,y,z)=4h$.
    By \cref{lem:diamequalshyp}, we know that there are exactly two pairs of distance~$4h$ and, hence, $\{w,x,y,z\}\subseteq A\cup B\cup\{u,v\}\cup M^*$.
    We claim that $|\{w,x,y,z\}\cap (M^*\cup \{u,v\})|\leq 2$.
    By \cref{lem:diamequalshyp}, we know that, out of~$w,x,y,z$, there are exactly two pairs 
    at distance~$4h$ and all other pairs have distance~$2h$.
    Assume that $|\{w,x,y,z\}\cap M^*\cup \{u,v\}|\geq 3$. 
    Then, at least two vertices are in~$P_v^A\cup P_v^B\cup \{v\}$ or in~$P_u^A\cup P_u^B\cup \{u\}$.
    Observe that any two vertices in~$P_v^A\cup P_v^B\cup \{v\}$ or in~$P_u^A\cup P_u^B\cup \{u\}$ are at distance smaller than~$2h$,
    but this contradicts the choice of the quadruple.
    It follows that $|\{w,x,y,z\}\cap M^*\cup \{u,v\}|\leq 2$, and w.l.o.g.\ let $w,x\in A\cup B$.
    As each vertex in~$A$ is at distance smaller than $3h$ to any vertex in $A\cup \{u,v\}\cup M^*$, it follows that the other vertex is in~$B$.
    Applying \cref{claim:orthog4h},
    we have that~$w$ and~$x$ are at distance~$4h$ if and only if~$\vec{w}$ and~$\vec{x}$ are orthogonal;
    hence, the statement of the lemma follows.
  \end{proof}
}

\section{Parameter Distance to Cographs}
We now describe a fixed-parameter linear-time algorithm for \textsc{Hyperbolicity}
parameterized by the vertex deletion distance~$k$ to cographs. A graph is a cograph if and
only if it is~$P_4$-free. 
Given a graph~$G$ we can determine in linear time whether it is
a cograph and return an induced~$P_4$ if this is not the case. This implies that
in~$O(k\cdot (m+n))$ time we can compute a set~$X\subseteq V$ of size at most~$4k$ such
that~$G-X$ is a cograph.

A further characterization is that a cograph can be obtained from graphs consisting of one single vertex
via unions and joins~\cite{BLS99}.
\begin{compactitem}
\item A \emph{union} of two graphs~$G_1=(V_1,E_1)$ and~$G_2=(V_2,E_2)$ is the graph~$(V_1\cup V_2,E_1\cup E_2)$. 
\item A \emph{join} of two graphs~$G_1=(V_1,E_1)$ and~$G_2=(V_2,E_2)$ is the graph~$(V_1\cup V_2,E_1\cup E_2\cup \{\{v_1,v_2\}|v_1\in V_1,v_2\in V_2\})$.
\end{compactitem}
The union of $t$ graphs and the join of~$t$ graphs are defined by taking successive unions
or joins, respectively, of the~$t$ graphs in an arbitrary order. Each cograph~$G$ can be
associated with a rooted cotree~$T_G$. The leaves of~$T_G$ are the vertices of~$V$. Each
internal node of~$T_G$ is labeled either as a union or join node. For node $v$ in~$T_G$,
let~$L(v)$ denote the leaves of the subtree rooted at~$v$. For a union node~$v$ with
children~$u_1,\ldots, u_t$, the graph~$G[L(v)]$ is the union of the
graphs~$G[L(u_i)]$,~$1\le i\le t$. For a join node~$v$ with children~$u_1,\ldots, u_t$, the
graph~$G[L(v)]$ is the join of the graphs~$G[L(u_i)]$,~$1\le i\le t$.

The cotree of a cograph can be computed in linear time~\cite{CPS85}. In a subroutine in
our algorithm for \textsc{Hyperbolicity} we need to solve the following variant of \textsc{Subgraph Isomorphism}.
  \problemdef{\textsc{Colored Induced Subgraph Isomorphism}}
  	{An undirected graph~$G=(V,E)$ with a vertex-coloring~$\gamma:V\to \mathbb{N}$ and an undirected graph~$H=(W,F)$, where $|W|=k$, with a vertex-coloring~$\chi:W\to \mathbb{N}$.}
	{Is there a vertex set~$S\subseteq V$ such that there is an isomorphism~$f$ from~$G[S]$ to~$H$ such that~$\gamma(v)=\chi(f(v))$ for all~$v\in S$?}
Informally, the condition that~$\gamma(v)=\chi(f(v))$ means that every vertex is mapped to
a vertex of the same color. We say that such an isomorphism \emph{respects the colorings}.
As shown by Damaschke~\cite{Dam90}, \textsc{Induced Subgraph Isomorphism} on cographs is
\NP-complete. Since this is the special case of \textsc{Colored Induced Subgraph
  Isomorphism} where all vertices in~$G$ and~$H$ have the same color, \textsc{Colored
  Induced Subgraph Isomorphism} is also \NP-complete (containment in \NP is obvious). In the
following, we show that on cographs \textsc{Colored Induced Subgraph Isomorphism} can be
solved by a linear-time fixed-parameter algorithm when the parameter is the order~$k$
of~$H$.

\begin{lemma}%
  \label{lemma:cisifpt}
	\textsc{Colored Induced Subgraph Isomorphism} can be solved in~$O(3^k(n+m))$ time in cographs.
\end{lemma}

{
  \begin{proof}
    We use dynamic programming on the cotree. Herein, we assume that for each internal
    node~$v$ there is an arbitrary (but fixed) ordering of its children; the~$i$th child
    of~$v$ is denoted~$c_i(v)$ and the set of leaves in the subtrees rooted at the first~$i$
    children of~$v$ is denoted~$L_i(v)$. We fill a three-dimensional table~$D$ with entries
    of the type~$D[v,i,X]$ where~$v$ is a node of the cotree with at least~$i$ children 
    and~$X\subseteq W$ is a subset of the vertices of the pattern~$H$. The entry~$D[v,i,X]$
    has value 1 if~$(G[L_i(v)], \gamma|_{L_{i}(v)}, H[X], \chi|_{X})$ 
    is a yes-instance of \textsc{Colored Induced Subgraph Isomorphism}, that is, there is a subgraph isomorphism
    from~$G[L_i(v)]$ to~$H[X]$ that respects the coloring. Otherwise, the entry has value~0. 
    Thus,~$D[v,\deg(v)-1,X]$  has value 1 if and only if there is an induced subgraph
    isomorphism from~$G[L_{i}(v)]$ to~$H[X]$. After the table is completely filled, the
    instance is a yes-instance if and only if~$D[r,\deg(r),W]$ has value~1 where~$r$ is the
    root of the cotree. We initialize the table for leaf vertices~$v$, 
    by setting~$D[v,0,X]=1$ if either~$X=\emptyset$ or~$X = \{u\}$ with~$\gamma(v) = \chi(u)$; otherwise~$D[v,0,X]=0$.

    For union nodes, the table~$D$ is filled by the following recurrence
    \begin{displaymath}
      D[v,i,X]=
      \begin{cases}
	1 & \exists X'\subseteq X: D[v,i-1,X']= D[c_{i}(v),\deg(c_{i}(v))-1,X\setminus X']=1\\
	&  \text{$\wedge$ there are no edges between~$X'$ and~$X\setminus X'$ in~$H$}\\
	0 & \text{otherwise}.
      \end{cases}
    \end{displaymath}
    For join nodes, the table~$D$ is filled by the following recurrence
    \begin{displaymath}
      D[v,i,X]=
      \begin{cases}
	1 & \exists X'\subseteq X: D[v,i-1,X']= D[c_{i}(v),\deg(c_{i}(v))-1,X\setminus X']=1\\
	&  \text{$\wedge$ every~$q\in X'$ is adjacent to every~$p\in X\setminus X'$ in~$H$}\\
	0 & \text{otherwise}. 
      \end{cases}
    \end{displaymath}
    The correctness of the recurrence can be seen as follows for the union nodes. 
    First assume there is a color-respecting induced subgraph isomorphism from~$G[L_i(v)]$
    to~$H[X]$. Then there is a set~$S\subseteq L_i(v)$ such that there is a
    color-respecting isomorphism~$f$ from~$G[S]$ to~$H[X]$. Let~$S':=S\cap L_{i-1}(v)$ be
    the set of vertices that are from~$S$ and from~$L_{i-1}(v)$ which implies
    that~$S\setminus S'=S\cap L(c_{i}(v))$. Since~$v$ is a union node, there are no edges
    between~$S'$ and~$S\setminus S'$ in~$G$. Let~$X':=f(S')$ and~$X\setminus X'=f(S\setminus
    S')$ denote the image of~$S'$ and~$S\setminus S'$, respectively. Since~$f$ is an
    isomorphism there are no edges between between~$X'$ and~$X\setminus X'$. Moreover, since
    restricting a color-respecting isomorphism~$f\colon G[S]\to H[X]$ to a subset~$S'$ gives a
    color-respecting isomorphism from~$f\colon G[S']\to H[f(S')]$ we have that~$D[v,i-1,X']$
    and~$D[c_{i}(v),\deg(c_{i}(v))-1,X\setminus X']$ have value 1. Therefore, there is a case such that
    the recurrence evaluates correctly to 1.

    Conversely, if the recurrence evaluates to 1, then the conditions in the recurrence (about the existence of~$X'$) imply a color-respecting induced subgraph isomorphism from~$G[L_i(V)]$ to~$H[X]$. 
    Therefore the table is filled correctly for union nodes. 
    The correctness of the recurrence for join nodes follows by symmetric arguments.

    The running time is bounded as follows. The cotree has size~$O(n+m)$ and thus, there
    are~$O((n+m)\cdot 2^k)$ entries in the table. For each~$X\subseteq W$, filling the
    entries of a particular table entry is done by considering all subsets of~$X$, thus the
    overall number of evaluations is~$O(3^k\cdot (n+m))$.  %
  \end{proof}
}
We now turn to the algorithm for \hyp{} on graphs that can be made into cographs by at
most~$k$ vertex deletions.

\problemdef{\textsc{Distance-Constrained 4-Tuple}} {An undirected graph~$G=(V,E)$ and six integers $d_{\{a,b\}}$, $d_{\{a,c\}}$, $d_{\{a,d\}}$, $d_{\{b,c\}}$, $d_{\{b,d\}}$, and~$d_{\{c,d\}}$.}
{Is there a set~$S\subseteq V$ of four vertices and a bijection~$f\colon S\to \{a,b,c,d\}$ such that for each~$x,y\in S$ we have~$\dist{x}{y} =d_{\{f(x),f(y)\}}$?}

\begin{lemma}%
  \label{lem:cograph-dc-4-tuple}
  \textsc{Distance-Constrained 4-Tuple} can be solved in~$O(4^{4k}\cdot k\cdot (n+m))$
  time if~$G-X$ is a cograph for some~$X\subseteq V$ of size~$k$.
\end{lemma}

{
  \begin{proof}
    Let~$G = (V,E)$ be the input graph and~$X \subseteq V$, $|X| \le k$, such that~$G-X$ is
    a cograph. Without loss of generality, let~$X=\{x_1,\ldots ,x_k\}$.

    In a preprocessing step, we classify the vertices of each connected component in~$G[V
    \setminus X]$ according to the length of shortest paths to vertices in~$X$ such that all
    internal vertices of the shortest path are in~$V\setminus X$.

    More precisely, for a vertex $v \in V \setminus X$ in a connected component~$C_v$
    of~$G-X$, the \emph{type}~$t_v$ of~$v$ is a length-$k$ vector containing the distance
    of~$v$ to each vertex~$x_i$ of~$X$ in $G[C_v \cup \{x_i\}]$.  That is, $t_v[i]$ equals
    the distance from~$v$ to~$x_i \in X$ within the graph~$G[C_v \cup \{x_i\}]$.  Since the
    diameter of~$G[C_v]$ is at most two, $t_v[i] \in \{1,2,3,\infty\}$. Therefore, the
    number of distinct types in~$G$ is at most~$4^k$. 
    For simplicity of notation, for every type $t$ we denote by $v_t$ an arbitrary vertex such that $t_v=t$.

    \begin{observation} \label{obs:dist-type}
      Let~$u$ and~$v$ be two vertices of the same type, that is,~$t_u=t_v$. 
      Then for each vertex~$w$ in~$G-(C_u\cup C_v)$, we have~$\dist{u}{w}=\dist{v}{w}$.
    \end{observation}

    \begin{observation}
      Given two vertices~$u$ and~$v$ such that~$C_u\neq C_v$, we can compute~$\dist{u}{v}$ in~$O(k)$ time when the distance between each~$u\in X$ and each~$v\in V\setminus X$ can be retrieved in~$O(1)$ time.
    \end{observation} 

    The dominating part of the running time of the preprocessing is the computation of the
    vertex types which can be performed in~$O(k\cdot (n+m))$ time as follows. Create a
    length-$k$ vector for each of the at most~$n$ vertices of~$V\setminus X$ and initially
    set all entries to~$\infty$. Then compute for each~$x_i\in X$, the graph~$G-(X\setminus
    \{x_i\})$ in~$O(n+m)$ time. In this graph perform a breadth-first search from~$x_i$ to
    compute the distances between~$x_i$ and each vertex~$v\in V\setminus X$ that is in the
    same connected component as~$x_i$. This distance is exactly the one in the graph $G[C_v
    \cup \{x_i\}]$. Thus, for all vertices that reach~$x_i$, the~$i$th entry in their type
    vector is updated. Afterwards, each vertex has the correct type vector.

    After this preprocessing, the algorithm proceeds as follows by restricting the choice of vertices for the 4-tuple.
    \begin{itemize}
    \item First, branch into all~$O(k^4)$ cases of taking a subset~$X'\subseteq X$ of size at most four. (We will assume that~$X'=S\cap X$.)
    \item 
	  For each such $X'\subseteq X$, branch into the different cases for the types of vertices in $S\setminus X'$. 
	  That is, consider all multisets~$M_T$ of size~$|S\setminus X'| = 4 - |X'|$ over the universe of all types.
	  (There are~$4^k$ types and thus at most $4^{4k}$ cases for each~$X'\subseteq X$.)
    \item For each such $X'\subseteq X$ and multiset $M_T$, branch into all cases of matching the vertices in~$\{a,b,c,d\}$ to the vertices in~$X$ and types in~$M_T$ (branch into all ``bijections''~$f$ between~$X' \cup M_T$ and~$\{a,b,c,d\}$).
    (There are at most $4!$ cases.)

    \item For each such branch, branch into the different possibilities to assign the types
      in~$M_T$ to connected components of~$G-X$. That is, create one branch for each
      partition of the multiset~$M_T$ and assume in this branch that two types are in the
      same connected component if and only if they are in the same set of the partition
      of~$M_T$. The current partition is called the \emph{component partition} of the branch.
    \end{itemize}
    We now check whether there is a solution to the \textsc{Distance-Constrained 4-Tuple}
    instance that fulfills the additional assumptions made in the above branches. To this end, for
    each pair of vertices~$x,y\in X'$, check 
    whether~$\dist{x}{y}=d_{\{f(x),f(y)\}}$. Now, for each vertex~$x\in X'$ and each
    type~$t\in M_T$, check whether $\dist{x}{v_t}=d_{\{f(x),f(v_t)\}}$, where~$v_t$ is
    an arbitrary vertex of type~$t$. Observe that this is possible since, by
    \cref{obs:dist-type} the distance between~$X$ and any vertex of type~$t$ is
    the same in~$G$. Next, for pair of types~$t,t'\in M_T$ such that the branch
    assumes that~$t$ and~$t'$ do not lie in the same connected component of~$G-X$, check whether
    $\dist{v_{t}}{v_{t'}}=d_{\{f(v_{t}),f(v_{t'})\}}$. Again, this is possible due to
    \cref{obs:dist-type}. 

    The remaining problem is thus to determine whether the types of~$M_T$ can be assigned to
    vertices in such a way that
    \begin{itemize}
    \item for each pair of types~$t,t'\in M_T$ the assigned vertices are in the
      same connected component of~$G-X$ if and only if it is constrained to be in the same
      type of connected component in the current branch,
    \item for each pair of types~$t,t'\in M_T$ such that their assigned
      vertices~$v_t$ and~$v_{t'}$ are constrained to be in the same connected component, we
      need to ensure that~$\dist{v_t}{v_{t'}}=d_{\{f(v_t),f(v_{t'})\}}$.
    \end{itemize}
    We solve this problem by a reduction to~\textsc{Colored Induced Subgraph
      Isomorphism}. Observe that, since~$G-X$ is a cograph, for each pair~$u$ and~$v$ of
    vertices in the same connected component of~$G-X$, the distance between~$u$ and~$v$ is 2
    if and only if they are not adjacent. With this observation, the reduction works as
    follows. Let~$j\le 4$ denote the number of distinct connected components of~$G-X$ that
    shall contain at least one type of~$M_T$. Now, for each connected component~$C$ of~$G-X$
    add one further vertex~$v_C$ by making it adjacent to all vertices of~$C$ and call the
    resulting graph~$G'$. Now color the vertices of~$G-X$ as follows. The additional
    vertices of each connected component receive the color~$0$. Next, for each vertex
    type~$t$ in~$G-X$ introduce one color and assign this color to each vertex of
    type~$t$. Call the vertices with color~$0$ the \emph{component-vertices} and all other
    vertices the \emph{type-vertices}. To complete the construction of the input instance,
    we build~$H$ as follows. Add~$j$ vertices of color~$0$. Then add a vertex for each
    type~$t$ of~$M_T$ and color it with the color corresponding to its type. As in~$G'$,
    call the vertices with color~$0$ \emph{component-vertices} and all other vertices
    \emph{type-vertices}. Add edges between the component-vertices and the type-vertices in
    such a way that every type-vertex is adjacent to one vertex of color 0 and two
    type-vertices are adjacent to the same color-0 vertex if and only if they are
    constrained to be in the same connected component. Finally, if two type-vertices are
    constrained to be in the same connected component and have distance 1 in~$G$, then add
    an edge between them, otherwise add no edge between them. This completes the
    construction of~$H$. The instance of \textsc{Colored Induced Subgraph Isomorphism}
    consists of~$G'$ and~$H$ and of the described coloring. We now claim that this instance
    is a yes-instance if and only if there is a solution to the \textsc{Distance-Constrained
      4-Tuple} instance that fulfills the constraints of the branch.

    If the instance has a solution, then the subgraph isomorphism~$\phi$ from~$G'[S]$ to~$H$
    corresponds to a selection of types from~$j$ connected components since~$H$ contains
    neighbors of~$j$ component-vertices. Moreover, in~$H$, and thus in~$G'[S]$,~every
    type-vertex is adjacent to exactly one component-vertex and thus the component-vertices
    define a partition of the type-vertices of~$G'[S]$ that is, due to the construction
    of~$H$, exactly the component partition of~$M_T$. 
    Selecting the type-vertices of~$S$ and assigning them to~$\{a,b,c,d\}$ as specified by~$f$ gives, together with the selected vertices of~$X$, a special 4-tuple~$Q$ 
    Observe that~$Q$ fulfills all constraints of the branch except for the conditions on the distances between the type-vertices of the same component. 
    Now for two vertices~$u$ and~$v$ of~$S$ in the same connected component of~$G-X$, the
    distance is 1 if they are adjacent and 2 otherwise. Due to the construction of~$H$, and
    the fact that~$\phi$ is an isomorphism, the distance is thus 1 if~$d_{\{f(u),f(v)\}}=1$
    and 2 if~$d_{\{f(u),f(v)\}}=2$. Thus, if the \textsc{Colored Induced Subgraph
      Isomorphism} instance is a yes-instance, so is the~\textsc{Distance-Constrained
      4-Tuple} instance. The converse direction follows by the same arguments.

    The running time can be seen as follows. The preprocessing can be performed in~$O(k(n+m))$
    time, as described above. Then, the number of branches is~$O(4^{4k})$: the only time
    when the number of created branches is not constant is when the types of the vertices in the 
    4-tuple are constrained or when the 4-tuple vertices are fixed to belong to~$X$. 
    In the worst case, we have~$X' = \{a,b,c,d\}\cap X=\emptyset$, that is, $S \subseteq V\setminus X$ and for all four vertices of~$S$ one has to branch in total into~$4^{4k}$ cases to fix the types. 
    In each branch, the algorithm first checks the
    conditions on all distances except for the distances between vertices of the same parts
    of the component partition. This can be done in~$O(k)$ time for each of these
    distance. Afterwards, the algorithm builds and solves the \textsc{Colored Induced
      Subgraph Isomorphism} in~$O(n+m)$ time. Altogether, this gives the claimed running
    time bound. %
  \end{proof}
}
The final step is to reduce~\hyp{} to \textsc{Distance-Constrained 4-Tuple}. This can be
done by creating~$O(k^4)$ instances of \textsc{Distance-Constrained 4-Tuple} as shown
below.
\begin{theorem}\label{thm:cograph-dist}
	\hyp can be solved in~$O(4^{4k}\cdot k^7\cdot (n+m))$ time, where $k$ is the vertex deletion distance of $G$ to cographs.
\end{theorem}

\begin{proof}
  Let~$G = (V,E)$ be the input graph and~$X \subseteq V$, $|X| \le k$, such that~$G-X$ is
  a cograph and observe that~$X$ can be computed in~$O(4^k\cdot (n+m))$ time. Since every
  connected component of~$G-X$ has diameter at most two, the maximum distance between any pair of
  vertices in the same component of~$G$ is at most~$4k+2$: 
  any shortest path between two vertices~$u$ and~$v$ visits at most~$k$ vertices in~$X$, 
  at most three vertices between
  every pair of vertices~$x$ and~$x'$ from~$X$ and at most three vertices before
  encountering the first vertex of~$X$ and at most three vertices before encountering the
  last vertex of~$X$.
  
  Consequently, for the 4-tuple~$(a,b,c,d)$ that maximizes~$\delta(a,b,c,d)$, there
  are~$O(k^6)$ possibilities for the pairwise distances between the four vertices. Thus,
  we may compute whether there is a 4-tuple such that $\delta(a,b,c,d)=\delta$ by checking
  for each of the~$O(k^6)$ many $6$-tuples of possible pairwise distances of four vertices
  in~$G$ whether there are~$4$ vertices in~$G$ with these six pairwise distances and
  whether this implies~$\delta(a,b,c,d)\ge \delta$. The latter check can be performed
  in~$O(1)$ time, and the first is equivalent to solving \textsc{Distance-Constrained
    4-Tuple} which can be done in~$O(4^{4k}\cdot k\cdot (n+m))$ time by
  \cref{lem:cograph-dc-4-tuple}. The overall running time follows. %
\end{proof}

\section{Reduction from 4-Independent Set}\label{sec:ind-set}

In this section, we provide a further relative lower bound for \textsc{Hyperbolicity}.
Specifically, we prove that, if the running time is measured in terms of~$n$,
then~\textsc{Hyperbolicity} is at least as hard as the problem of finding an independent set of size four in a graph.
The currently best running time for this problem is~$O(n^{3.257})$~\cite{EG04,WWWY15}.
Hence, any improvement on the running time of~\textsc{Hyperbolicity} which breaks this bound (e.g., an algorithm running in $o(n^3)$ time),
would also yield a substantial improvement for the~\textsc{4-Independent Set} problem.

To this end, we reduce from a 4-partite (or 4-colored) variant of the \textsc{Independent Set} problem.
The standard reduction from \textsc{Independent Set} to \textsc{Multicolored Independent Set}
shows that this 4-colored variant has the same asymptotic running time lower bound as \textsc{4-Independent Set}.

\begin{theorem}%
  \label{thm:4is}
 Any algorithm solving \hyp in $O(n^c)$~time for some constant~$c$ yields an $O(n^c)$-time algorithm solving \textsc{4-Independent Set}.
\end{theorem}

{
  \begin{proof}
  Let~$G=(V=V_1\uplus V_2\uplus V_3\uplus V_4,E)$ be an instance of the \textsc{4-Colored-Independent Set} problem.
  Assume an arbitrary order on the vertices of $V_i$, that is, $V_i=\{v_1^i,\ldots,v_{n_i}^i\}$, where $n_i;=|V_i|$, for each $1\leq i\leq 4$.
  We construct a graph~$G'$, initially being the empty graph, as follows (we refer to~\cref{fig:4isreduc} for an illustration).
  \begin{figure}[t!]
    \centering
    \begin{tikzpicture}[draw=black!75, scale=0.75,-]
	\tikzstyle{vone}=[circle,fill=white,draw=black!80,minimum size=25pt,inner sep=0pt]
	\tikzstyle{vtwo}=[circle,fill=white,draw=black!80,minimum size=5pt,inner sep=0pt]
	\tikzstyle{vthree}=[circle,fill=white,draw=black!80,minimum size=20pt,inner sep=0pt]
	\tikzstyle{vfour}=[circle,fill=white,draw=black!80,minimum size=35pt,inner sep=0pt]
	  \def\xsh{1.25};
	  \def\ysh{1};
	\foreach [count=\i] \pos / \text / \style in {
	  {(2,10*\ysh)}/$X_1$/vone,
	  {(8*\xsh,10*\ysh)}/$X_2$/vone,
	  {(8*\xsh,4)}/$X_3$/vone,
	  {(2,4)}/$X_4$/vone,
	  {(4*\xsh,11*\ysh)}//vtwo,
	  {(6*\xsh,11*\ysh)}//vtwo,
	  {(9*\xsh,8*\ysh)}//vtwo,
	  {(9*\xsh,6*\ysh)}//vtwo,
	  {(6*\xsh,3)}//vtwo,
	  {(4*\xsh,3)}//vtwo,
	  {(1,6*\ysh)}//vtwo,
	  {(1,8*\ysh)}//vtwo,
	  {(5*\xsh,10*\ysh)}/$X_1' $/vthree,
	  {(8*\xsh,7*\ysh)}/$X_2'$/vthree,
	  {(5*\xsh,4)}/$X_3' $/vthree,
	  {(2,7*\ysh)}/$X_4' $/vthree,
	  {(-1*\xsh,8*\ysh)}/$Y_1$/vthree,
	  {(6.5*\xsh,1*\ysh)}/$ Z_1$/vthree,
	  {(10.75*\xsh,8*\ysh)}/$Y_2$/vthree,
	  {(3.5*\xsh,1*\ysh)}/$Z_2$/vthree}
	{
	  \node[\style] (V\i) at \pos {\text};
	}

	  \node at (V5)[label=90:{\footnotesize $u_{1,2}$}]{};
	  \node at (V6)[label=90:{\footnotesize$u_{2,1}$}]{};
	  \node at (V7)[label=0:{\footnotesize$u_{2,3}$}]{};
	  \node at (V8)[label=0:{\footnotesize$u_{3,2}$}]{};
	  \node at (V9)[label=-90:{\footnotesize$u_{3,4}$}]{};
	  \node at (V10)[label=-90:{\footnotesize$u_{4,3}$}]{};
	  \node at (V11)[label=180:{\footnotesize$u_{4,1}$}]{};
	  \node at (V12)[label=180:{\footnotesize$u_{1,4}$}]{};

	  \foreach \x in {0.1,0.2,...,0.5} {
		  \draw[-] (V5) to ($(V1)+(0.1*\x+0.2,\x)$);
		  \draw[-] (V6) to ($(V2)+(0.1*\x+0.1,\x)$);
				  \draw[-] (V7) to ($(V2)+(0.75*\x-0.4,1.2*\x)$);
				  \draw[-] (V8) to ($(V3)+(0.75*\x-0.4,-0.5*\x)$);
				  \draw[-] (V9) to ($(V3)+(-0.1*\x-0.2,-\x)$);
		  \draw[-] (V10) to ($(V4)+(-0.4*\x+0.3,-\x+0.1)$);
				  \draw[-] (V11) to ($(V4)+(0.75*\x-0.4,1.2*\x)$);
				  \draw[-] (V12) to ($(V1)+(0.75*\x-0.4,-0.5*\x)$);
	  }
	\foreach \i / \j in {5/6,7/8,9/10,11/12} {
	  \path[] (V\i) edge (V\j) ;
	}

	\foreach \i / \j in {1/13,2/14,3/15,4/16} {
	  \path[very thick] (V\i) edge (V\j) ;
	}

	  \draw[-] (V13) to node[below,scale=0.66]{$\not\in E$}(V2);
		  \draw[-] (V14) to node[left,scale=0.66]{$\not\in E$}(V3);
		  \draw[-] (V15) to node[above,scale=0.66]{$\not\in E$}(V4);
		  \draw[-] (V16) to node[right,scale=0.66]{$\not\in E$}(V1);

	\foreach \i / \j in {1/17} {
	  \path[very thick] (V\i) edge (V\j) ;
	}

	\foreach \i / \j in {2/19} {
	  \path[very thick] (V\i) edge (V\j) ;
	}
	  \draw[very thick] (V19) to [out=-90,in=-45](V20);
	  \draw[very thick] (V17) to [out=-90,in=225](V18);
	  \draw[-] (V20) to node[right,scale=0.66]{$\in E$}(V4);
	  \draw[-] (V18) to node[right,scale=0.66]{$\in E$}(V3);

	\foreach [count=\y] \pos / \text / \style in {
	  {(2,10*\ysh)}/$X_1$/vone,
	  {(8*\xsh,10*\ysh)}/$X_2$/vone,
	  {(8*\xsh,4)}/$X_3$/vone,
	  {(2,4)}/$X_4$/vone,
	  {(4*\xsh,11*\ysh)}//vtwo,
	  {(6*\xsh,11*\ysh)}//vtwo,
	  {(9*\xsh,8*\ysh)}//vtwo,
	  {(9*\xsh,6*\ysh)}//vtwo,
	  {(6*\xsh,3)}//vtwo,
	  {(4*\xsh,3)}//vtwo,
	  {(1,6*\ysh)}//vtwo,
	  {(1,8*\ysh)}//vtwo,
	  {(5*\xsh,10*\ysh)}/$X_1' $/vthree,
	  {(8*\xsh,7*\ysh)}/$X_2'$/vthree,
	  {(5*\xsh,4)}/$X_3' $/vthree,
	  {(2,7*\ysh)}/$X_4' $/vthree,
	  {(-1*\xsh,8*\ysh)}/$Y_1$/vthree,
	  {(6.5*\xsh,1*\ysh)}/$ Z_1$/vthree,
	  {(10.75*\xsh,8*\ysh)}/$Y_2$/vthree,
	  {(3.5*\xsh,1*\ysh)}/$Z_2$/vthree}
	  {
	  \node[\style] (V\y) at \pos {\text};
	}

	  \node (w) at (5*\xsh,7*\ysh)[circle,draw,scale=0.5,label=45:{$w$}]{};
	  \foreach \x in {-0.4,-0.3,...,0.4}{
		  \draw[-] (w) to ($(w)+(1.3*\x,1.1*\ysh)$);
		  \draw[-] (w) to ($(w)+(1.1*\xsh,1.3*\x)$);
		  \draw[-] (w) to ($(w)+(1.3*\x,-1.1*\ysh)$);
		  \draw[-] (w) to ($(w)+(-1.1*\xsh,1.3*\x)$);
	  }
      \end{tikzpicture}
      \caption{An illustrative sketch of the graph constructed in the proof of \cref{thm:4is}. 
	  Encircled vertices correspond to cliques.
	  A thick edge represents a matching between copy-vertices.
	  An edge labeled ``$\in E$'' (``$\not\in E$'') represent incidences corresponding to present (not present) edges in the original graph.
	  The vertex $w$ is incident with all vertices beside those in the $X_i$s, $1\leq i\leq 4$.}
      \label{fig:4isreduc}
  \end{figure}
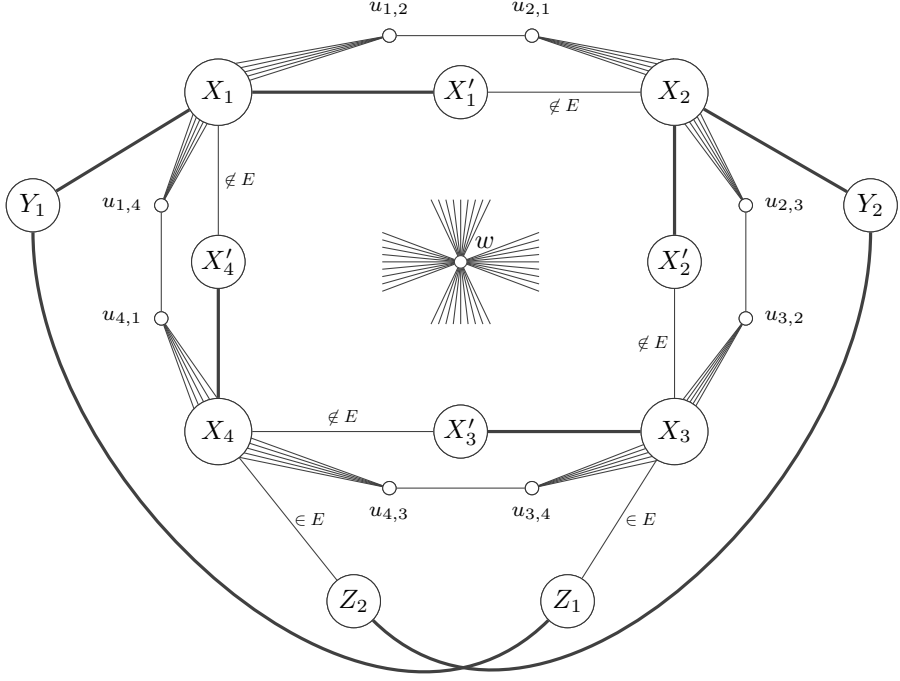
  \begin{compactitem}

  \item
  Add the vertex sets~$X_1$,~$X_2$,~$X_3$, and~$X_4$, where $X_i=\{x_1^i,\ldots,x_{n_i}^i\}$, $1\leq i\leq 4$.
  We say $x_j^i$ corresponds to the vertex $v_j^i\in V_i$ in~$V$, for each $1\leq i\leq 4$, $1\leq j\leq n_i$.
  Introduce a copy~$X'_i$ of each~$X_i$
  and further copies~$Y_1, Z_1$ of~$X_1$ and $Y_2, Z_2$ of~$X_2$. 
  Make each~$X_i$ and each copy of each~$X_i$ a clique.
  We say that the~$j$th vertex of some copy of~$X_i$ \emph{corresponds to the~$j$th vertex} of~$X_i$ and hence corresponds to the $j$th vertex in~$V_i$.

  \item
  For each vertex in~$X_i$ introduce an edge to its corresponding vertex in~$X'_i$.
    
  \item
  For $i\in\{1,2,3\}$, introduce an edge between a vertex in~$X'_i$ and a vertex in~$X_{i+1}$
  if their corresponding vertices in~$V$ are \emph{not} adjacent in~$G$. 
  Introduce edges between vertices in~$X_4'$ and~$X_1$ analogously.

  \item
  For $i\in\{1,2\}$, introduce edges for corresponding vertices between~$X_i$ and~$Y_i$, and between~$Y_i$ and~$Z_i$.

  \item
  For $i\in\{1.2\}$. introduce an edge between a vertex in~$Z_i$ and a vertex in~$X_{i+2}$
  if their corresponding vertices in~$V$ \emph{are} adjacent in~$G$.

  \item
  Introduce a set~$U:=\{u_{1,2}^1,u_{1,2}^2,u_{2,3}^2,u_{2,3}^3,u_{3,4}^3,u_{3,4}^4,u_{4,1}^4,u_{4,1}^1\}$
  of eight further vertices and call the vertices in~$U$ the \emph{connection vertices}.

  \item
  Introduce the edges~$\{u_{i,j}^i,u_{i,j}^j\}$, and connect each vertex in~$X_i$ with~$u^i_{i,j}$ and~$X_j$ with~$u^j_{i,j}$ via an edge.

  \item
  Finally, add the vertex~$w$ and connect~$w$ via an edge with all vertices except the vertices in~$X_i$, $1\leq i \leq 4$.

  \end{compactitem}

  This finishes the construction of~$G'=(V',E')$.
  Observe that $|V(G')|=2\cdot|V(G)|+2\cdot(n_1+n_2)+9$.
  Moreover, observe that the diameter of~$G'$ is four.
  To see this, observe that~$w$ has distance at most two to each vertex in~$G'$.
  Assuming that there exist at least one pair of vertices $a \in V_1$ and $c \in V_3$ such that $\{a, c\} \notin E$ (as otherwise~$G$ is a trivial no-instance),
  the distance between~$a$ and~$c$ is exactly~$4$.
  We prove that~$G'$ is 4-hyperbolic if and only if~$G$ has an independent set of size 4.

  \smallskip\noindent
  \raproof \quad
  Let~$\{a,b,c,d\}$ be a colored-independent set of size four in~$G$, and let without loss of generality~$a\in V_1$,~$b\in V_2$, $c\in V_3$, and~$d\in V_4$.
  Let $a',b',c',d'\in V(G')$ with $a'\in X_1$,~$b'\in X_2$, $c'\in X_3$, and~$d'\in X_4$ be the corresponding vertices in~$X:=X_1\cup X_2\cup X_3\cup X_4$.
  We show that~$\delta(a',b',c',d')=4$.

  First, we show that~$\dist{a'}{b'} = 2$.
  As no vertex in~$X_1$ is adjacent to any vertex in~$X_2$, we have~$\dist{a'}{b'} \geq 2$.
  As~$a$ and~$b$ are not adjacent in~$G$, they have a common neighbor in $X_1'$ by construction of~$G'$. 
  It follows that~$\dist{a'}{b'} = 2$.
  By a symmetric argument, we conclude that
  $\dist{b'}{c'} = \dist{c'}{d'} = \dist{d'}{a'} = 2$.

  Further, we show that~$\dist{a'}{c'}=4$.
  As each vertex in~$G'$ is at distance two to vertex~$w$, it follows that~$\dist{a'}{c'}\leq 4$.
  Moreover,
  all the neighbors of~$a'$ are in~$X'_1\cup Y_1\cup X'_4\cup \{u_{1,2}^1,u_{4,1}^1\}$
  and all the neighbors of~$c'$ are in~$Z_1\cup X'_2\cup X'_3\cup \{u_{2,3}^3,u_{3,4}^3\}$.
  Thus, to have a distance of at most three,
  a neighbor of~$a'$ must be adjacent to a neighbor of~$c'$.
  By construction, this is only possible if the unique neighbor~$a'_Y$ of~$a'$ in~$Y_1$ is adjacent to a neighbor of~$c'$ in~$Z_1$.
  The unique neighbor~$a'_Z\in Z_1$ of~$a'_Y\in Y_1$ is,
  however, not adjacent to~$c'$ since~$a$ and~$c$ are not adjacent in~$G$.
  It follows that~$\dist{a'}{c'}=4$.
  By a symmetric argument, we conclude that~$\dist{b'}{d'}=4$.

  Finally, altogether we have~$\delta(a',b',c',d')=(4+4)-(2+2)=4$.

  \smallskip\noindent
  \laproof \quad
  Let~$S:=\{a',b',c',d'\}\subseteq V(G')$ be a vertex set such that~$\delta(a',b',c',d')=4$.
  We show that these vertices correspond to four vertices in~$G$ forming a colored independent set in~$G$.
  By~\cref{lem:hyp-distance-bounded}, we have $4=\delta(a',b',c',d') \le 2 \cdot \min_{u \neq v \in S} \{\dist{u}{v}\}$, and hence no two vertices in~$S$ are adjacent.

  Now, assume without loss of generality that~$\dist{a'}{c'}+\dist{b'}{d'}$ is the largest sum among the distances.
  Since the diameter of~$G'$ is four, we have that~$\dist{a'}{c'}+\dist{b'}{d'}\le 8$.
  Moreover, all other distances are at least two,
  since~$S$ forms an independent set in~$G'$.
  Then, since~$G'$ is 4-hyperbolic,
  this implies that~$\dist{a'}{c'}=\dist{b'}{d'}=4$ and~$\dist{a'}{b'}=\dist{b'}{c'}=\dist{c'}{d'}=\dist{d'}{a'}=2$.

  This already implies that~$w\notin S$ as~$w$ has distance at most two to all other vertices in~$G'$.
  Moreover, since all vertices in~$V'\setminus X$ are adjacent to~$w$,
  each of them is at distance at most three to all other vertices in~$G'$.
  Thus,~$S\subseteq X$, but as~$S$ forms an independent set in~$G'$, there are no two vertices of~$X_i$, $1\leq i\leq 4$, in~$S$ (recall each~$X_i$ forms a clique in~$G'$).
  Let $a\in V_1$,~$b\in V_2$,~$c\in V_3$, and~$d\in V_4$ be the vertices in~$G$ corresponding to~$a'$,~$b'$,~$d'$, and~$d'$, respectively.
  Assume without loss of generality that~$a'\in X_1$,~$b'\in X_2$,~$c'\in X_3$, and~$d'\in X_4$.
  Finally, by the construction of~$G'$ together with~$\dist{a'}{b'}=\dist{b'}{c'}=\dist{c'}{d'}=\dist{d'}{a'}=2$ and~$\dist{a'}{c'}=\dist{b'}{d'}=4$,
  it follows that $\{a,b,c,d\}$ forms a colored-independent set in~$G$.
  \end{proof}
}

\section{Conclusion}
To efficiently compute the hyperbolicity number, parameterization sometimes 
may help. In this respect, perhaps our practically most promising results
relate to the $O(k^4(n+m))$ running times (for the parameters 
covering path number 
and feedback edge number, see \cref{tab:results})---note that 
they clearly improve on the standard algorithm when $k=o(n^{1/4})$. Moreover, 
the linear-time data 
reduction rules we presented may be of independent practical interest.
On the lower bound side, together with the work of Abboud et al.~\cite{AWW16}
our SETH-based lower bound with respect to the parameter vertex cover number 
is among few known ``exponential lower bounds'' for a 
polynomial-time solvable problem.

As to future work, we particularly point to the following open questions.
First, we left open whether there is a linear-time FPT algorithm 
exploiting the parameter feedback vertex number for computing the 
hyperbolicity number.
Second, for parameter vertex cover number we have an SETH-based exponential 
lower bound for the parameter function in any linear-time FPT~algorithm. 
This does not imply that it is impossible to achieve a polynomial parameter 
dependence when asking for algorithms with running time factor $O(n^2)$ or~$O(n^3)$.

\bibliographystyle{abbrvnat}
\bibliography{hyper-arxiv2017-short} %

\end{document}